\documentclass[conference]{IEEEtran}
\IEEEoverridecommandlockouts
\usepackage{cite}
\usepackage{amsmath,amssymb,amsfonts,amsthm}
\usepackage{algorithmic}
\usepackage{textcomp}
\def\BibTeX{{\rm B\kern-.05em{\sc i\kern-.025em b}\kern-.08em
    T\kern-.1667em\lower.7ex\hbox{E}\kern-.125emX}}

\usepackage{mathrsfs}
\usepackage{mathabx}
\usepackage[hidelinks]{hyperref} 
\usepackage{color}
\usepackage{graphicx}
\usepackage{graphbox}
\usepackage[dvipsnames]{xcolor}
\usepackage{tikz}
\usepackage{tikz-cd}
\usetikzlibrary{shapes.geometric, arrows, positioning, decorations.pathmorphing,shapes}
\usetikzlibrary{cd}

\newcommand\Ter[1]{\textcolor{ForestGreen}{\textrm{Teresa: }#1}}
\newcommand\ST[1]{\textcolor{teal}{[ST: #1]}}

\newcommand{\compl}{\textrm{c}} 
\newcommand{\ballop}[2]{B_{#1}(#2)} 
\newcommand{\diam}{\textrm{diam}} 
\newcommand{\reach}{\ensuremath{\textrm{reach}}} 
\newcommand{\mreach}[1]{\ensuremath{\reach(\partial #1)}}
\newcommand{\dt}[1]{d^{\mp}_{#1}} 
\newcommand{\dtonesided}[1]{d_{#1}} 
\newcommand{\tcx}{X(r,t)} 
\newcommand{\DSEDT}[1]{\textrm{D}[#1]} 
\newcommand{\tg}{D_r} 
\newcommand{\dH}{d_{\mathcal{H}}} 
\newcommand{\leash}[2]{\textrm{leash}_{#1}(#2)} 
\newcommand{\mleash}[2]{l_{#1}(#2)} 
\newcommand{\erosion}[2]{\textrm{E}_{#2}(#1)} 
\newcommand{\mm}{\ensuremath{m_{\rho_r}}} 

\theoremstyle{plain}
\newtheorem{lemma}{Lemma}[section]
 \newtheorem{corollary}[lemma]{Corollary}
 \newtheorem{proposition}[lemma]{Proposition}
 \newtheorem{theorem}[lemma]{Theorem}
 \newtheorem{example}[lemma]{Example}
\theoremstyle{definition}
\newtheorem{definition}[lemma]{Definition}
 \newtheorem{remark}[lemma]{Remark}

\newcommand{\R}{\ensuremath{\mathbb{R}}}
\newcommand{\Z}{\ensuremath{\mathbb{Z}}}

\DeclareMathOperator{\cl}{cl}
\DeclareMathOperator{\inter}{int}

\DeclareMathOperator{\dB}{d_B}
   
\DeclareMathOperator{\cf}{f} 
\DeclareMathOperator{\dfr}{f_r} 

\newcommand{\PD}[1]{\ensuremath{\mathrm{PD}\!\left({#1}\right)}}
\newcommand{\kPD}[2]{\ensuremath{\mathrm{PD}_{#1}\!\left({#2}\right)}}
\newcommand{\norm}[1]{\ensuremath{\left\lVert {#1} \right\rVert}} 
\newcommand{\norminf}[1]{\ensuremath{\left\lVert {#1} \right\rVert_\infty}} 

\newcommand{\absolute}[1]{\ensuremath{|{#1}|}} 

\begin{document}
\title{The Impact of Changes in Resolution on the Persistent Homology of Images\\
\thanks{The authors thank the Mathematical Sciences Institute at ANU, the US National Science Foundation through the award CCF-1841455, the Australian Mathematical Sciences Institute and the Association for Women in Mathematics for funding the second Workshop for Women in Computational Topology in July 2019 where this project began. The project reached completion thanks to funding from the MSRI Summer Research in Mathematics program awarded in 2020. 
Teresa Heiss has received funding from the ERC under the Horizon 2020 programme (No.\ 788183).
We are also grateful to Kritika Singhal for her input to the first part of the project. The authors want to thank Mathijs Wintraecken for the idea of tightening the bounds when the Hausdorff distance is below the reach, i.e.\ the idea for Lemma~\ref{lemma:app:remark:tighteningCS-E-H} in the appendix.}
}

\author{\IEEEauthorblockN{Teresa Heiss\IEEEauthorrefmark{1}, 
Sarah Tymochko\IEEEauthorrefmark{2},
Brittany Story\IEEEauthorrefmark{3}, 
Ad\'{e}lie Garin\IEEEauthorrefmark{4}, 
Hoa Bui\IEEEauthorrefmark{5}, 
Bea Bleile\IEEEauthorrefmark{6}, and 
Vanessa Robins\IEEEauthorrefmark{7}
}
\IEEEauthorblockA{\IEEEauthorrefmark{1}Institute of Science and Technology (IST) Austria, Klosterneuburg, Austria}
\IEEEauthorblockA{\IEEEauthorrefmark{2}Michigan State University, Department of Computational Mathematics, Science and Engineering, East Lansing, MI, USA}
\IEEEauthorblockA{\IEEEauthorrefmark{3}Colorado State University, Department of Mathematics, Fort Collins, CO, USA}
\IEEEauthorblockA{\IEEEauthorrefmark{4} \'{E}cole Polytechnique Fédérale de Lausanne (EPFL), Switzerland }
\IEEEauthorblockA{\IEEEauthorrefmark{5}{School of Electrical Engineering, Computing and Mathematical Sciences, Curtin University, Australia}}
\IEEEauthorblockA{\IEEEauthorrefmark{6}School of Science and Technology, University of New England, Armidale, Australia}
\IEEEauthorblockA{\IEEEauthorrefmark{7}Research School of Physics, The Australian National University, Canberra, Australia}
}

\maketitle

\begin{abstract}
Digital images enable quantitative analysis of material properties at micro and macro length scales, but choosing an appropriate resolution when acquiring the image is challenging. 
A high resolution means longer image acquisition and larger data requirements for a given sample, but if the resolution is too low, significant information may be lost. 
This paper studies the impact of changes in resolution on persistent homology, a tool from topological data analysis that provides a signature of structure in an image across all length scales. 
Given prior information about a function, the geometry of an object, or its density distribution at a given resolution, we provide methods to select the coarsest resolution yielding results within an acceptable tolerance. 
We present numerical case studies for an illustrative synthetic example and samples from porous materials where the theoretical bounds are unknown. 
\end{abstract}


\section{Introduction}\label{sec:intro}

Engineers know that macro-scale properties, such as the strength of a composite material and the rate of fluid flow through a porous solid, depend crucially on the microscopic geometric and topological structure of the material. 
In practice, micro-CT x-ray images of porous materials provide detailed three-dimensional maps of spatial structure for the different phases present --- typically summarised as solid grains and fluid-filled pores~\cite{wildenschild_x-ray_2013}.
The size, shape, and connectivity of the pore space are known to dictate physical properties of fluid flow such as permeability and trapping capacity~\cite{herring_effect_2013,herring_topological_2019}. 
This structure is exactly the type of information captured by a key tool of Topological Data Analysis (TDA), namely persistent homology. 

Simply put, persistent homology is the study of connected components and holes in a sequence of nested shapes~\cite{perssurvey}.
In image analysis, these nested shapes are obtained by thresholding a real-valued function defined on the image domain. 
For example, to characterise the geometry of a porous material an appropriate function is the signed Euclidean distance transform derived from the binary map of solid and pore~\cite{cub1,delgado-friedrichs_morse_2014}. 
The persistence diagrams computed from such a distance transform then quantify the sizes, shapes, and connectivity of all the individual pores and grains.  
A key property of persistent homology is its stability with respect to perturbations in the input function~\cite{cohen2007stability}.   
Topological features with small persistence can be filtered out and may correlate with noise (although this depends on the type of noise and how the input function is defined and discretized). 
This stability does not hold for the regular homology of a shape because every topological feature is counted with equal `weight' no matter its size.  

Capturing  both the  micro-  and  macro-scale  properties of a material using micro-CT is a trade-off  between  accuracy at the micro-scale and the computational expense associated with the large image dimensions required to represent a macro-scale sample.  
The main focus of this paper is to establish both theoretical and numerical results that bound the difference between persistence diagrams computed from images of an object taken at different resolutions.
%
With prior knowledge of the object these bounds can guide the selection of the coarsest resolution that will yield a persistence diagram within an acceptable distance of the ground-truth.
When high resolution images lead to ``big data'' problems these can be scaled back to a more manageable size by reducing the image resolution or subsampling to a lower resolution. 
As persistent homology is very expensive to compute on high resolution imagery, studying large datasets of imagery is nearly impossible; being able to select a coarser resolution can alleviate this issue and allow for more data to be analyzed.

After covering the mathematical background in Section~\ref{sec:math_background}, we state results in Section~\ref{sec:PH_grayscale} comparing real-valued functions (not necessarily continuous) with digital approximations for different resolutions. 
Since we define the digital approximation of $f$ by taking average values within each voxel, these results involve fairly straightforward bounds on the variation of $f$ and an application of the stability theorem for persistence diagrams. 

In Section~\ref{sec:PH_distance_transform} we focus on a particular application where the ground-truth is the continuous signed Euclidean distance transform (CSEDT) derived from an object $X \subset \R^d$.  
We compare this function with the discrete signed Euclidean distance transform (DSEDT) of a digital approximation to $X$. 
This is more complex than simply applying the results of Section~\ref{sec:PH_grayscale} to the CSEDT of $X$, because the digital approximation to $X$ may result in small components being lost, and this in turn may lead to a large difference between the CSEDT and the DSEDT. 
The results and bounds obtained in this section are illustrated with simple examples to show their relevance. 


In Section~\ref{sec:theory_not_apply} we present a numerical case study of a synthetic image with structure on different length scales to illustrate circumstances where the theoretical bounds are unknown. Finally, we study some porous materials in Section~\ref{sec:matsci}.
These experiments show that the actual data behaves even better than the bounds derived in Section~\ref{sec:PH_distance_transform} suggest.

\subsection{Related Work}
\label{sec:related_work}

There are several results on point-cloud approximations of manifolds in the context of geometric triangulations built from randomly sampled points near a manifold embedded in $\R^d$, \cite{amenta_simple_2000,niyogi_finding_2008,chazal_sampling_2009,VR_approx,kim2020homotopy}. 
Most of these results start with the assumption that the point-cloud approximation is close to the original object $X$ in the Hausdorff distance. 
In contrast, the main challenge in Section~\ref{sec:PH_distance_transform} is to establish that the digital approximation is close to the object $X$ in the Hausdorff distance. 
Once this is achieved, we apply  Lemma~\ref{lemma:bound_with_suprema} --- an analogous result to those for point-clouds, but in the setting of digital images and extended to the signed 
distance transforms. 

An earlier paper working with the persistent homology of digital images~\cite{bendich_computing_2010} uses an adaptive grid derived from an oct-tree data structure to reduce data size.  
Those authors bound distance between persistence diagrams of the original high-resolution image and their oct-tree approximation, similar in spirit to the results in our Section~\ref{sec:PH_grayscale} and the numerical case study in Section~\ref{sec:theory_not_apply}, but using the adaptive grid in the approximation rather than a single coarser voxel size. 

Related work in~\cite{dlotko_rigorous_2018} starts with a known continuous function, $f$, builds an adaptive rectangular subdivision of the domain and uses rigorous computer arithmetic to guarantee their piecewise constant approximation is within $\varepsilon$ of $f$ thus guaranteeing the same bound on the bottleneck distances between persistence diagrams. 

Papers that have studied how estimates of material properties from micro-CT images change with image  resolution include~\cite{botha_mapping_2016,huang_effect_2021}. 


\section{Mathematical Background}\label{sec:math_background}

\subsection{Persistent Homology}

A summary of the mathematical formalism and results of persistent homology is available in~\cite{perssurvey}, and we very briefly cover the basic idea and necessary notation in this section.  

Persistent homology is a tool from TDA that quantifies shape features that appear and then disappear during a \emph{filtration}, i.e., a sequential growth process indexed by a real parameter. 
The filtrations in this paper are obtained as sublevel sets of real-valued functions. 
Given $f\colon\R^d \to \R$, these sublevel sets are $X_\delta = f^{-1}\left(-\infty,\delta \right]$. 
In three-dimensions, the shape features of $X_\delta$ are the numbers of components, tunnels, and voids, i.e., its homology classes in dimensions $k=0,1,2$. Each feature has an associated interval $\left[b,d\right)$ that defines the range of filtration parameter values for which that feature persists.  
The persistence intervals for a filtration defined by $f$ are collected in \emph{persistence diagrams} which are multisets of points for each dimension of homology, $\kPD{k}{f} = \{(b_i,d_i)\}_{i=1}^{n_k}$.

Let $f, f'$ be two functions with persistence diagrams $\kPD{k}{f}$ and $\kPD{k}{f'}$. 
The \emph{bottleneck distance} between $\kPD{k}{f}$ and $\kPD{k}{f'}$ is defined using matchings $\gamma$ on points in the two diagrams, 
\[\dB(\kPD{k}{f},\kPD{k}{f'}) = \min_{\gamma} \max_x \lVert x - \gamma(x) \rVert_\infty.\] 
Note that a matching, $\gamma$, is permitted to pair any point $(b, d)$ in either diagram with a point on the diagonal in the other diagram.   
The stability theorem of persistent homology~\cite{cohen2007stability} now tells us that if $f$ and $f'$ are tame functions that are point-wise close then their persistence diagrams are also: 
\[ \dB(\kPD{k}{f},\kPD{k}{f'}) \leq \lVert f - f' \rVert_\infty. \]

\subsection{Digital Images as Discretized Functions} \label{subsec:images_cont_func}

There is a variety of different methods to define digital approximations to algebraic functions whose domain is $\R^d$. 
Here we consider a tame $\mu$-integrable function $\cf \colon \R^d \to \R$ (where $\mu$ denotes Lebesgue measure) and take local averages over each voxel. 

Let $r$ be the desired spacing for the digital grid and define a voxel $\sigma$ to be an open $d$-cube of side length $r$, and volume $\mu(\sigma) = r^d$. 
Voxels are indexed by the integer grid $\Z^d$ so that $\sigma(i) = \sigma(i_1,\ldots,i_d) \subset \R^d$ is the product of $d$ open intervals $(i_k r, (i_k +1)r)$. 
The \emph{digital approximation} to $\cf$ is the function 
\[ 
\dfr: \R^d \to \R, 
\]
which is a piecewise constant function defined on voxels $\sigma(i)$ with  
\[
\dfr(x) = r^{-d} \int_{\sigma(i)} \cf d\mu, \quad \text{if } x \in \sigma(i). 
\]
On the voxel faces, we define $\dfr(x)$ to be the minimum value taken on voxels $\sigma$ with $x \in \partial \sigma$.  

If we restrict $\dfr$ to the voxel centers indexed by integers $i \in \Z^d$, this discrete function is usually referred to as a \emph{grayscale digital image}.  
The above definition of $\dfr$ then agrees with the $T$-construction for digital images defined in~\cite{duality}.  
The cubical complex and its sublevel sets are equivalent to using the indirect adjacency of digital grids. 


In Section~\ref{sec:PH_grayscale} we use a bound on the differences $\lvert \cf(x) - \dfr(x) \rvert$ to obtain a corresponding bound on the bottleneck distance of their persistence diagrams.

\subsection{Digital Images of Objects and their Distance Transforms} \label{subsec:images_object}

Let $X\subset \R^d$ be a compact subset representing the solid object we approximate with a digital image.
By a solid object, we mean $\cl(\inter(X)) = X$.  
Typically $d=2,3$, and we can shift $X$ so that it is a subset of a rectangular prism, $R$.

A \emph{digital approximation to $X$} is defined using a discrete function $\rho_r: \Z^d \to [0,1]$ that quantifies the proportion of space occupied by $X$ in each voxel $\sigma(i)$ with 
\[ \rho_r(i) = \frac{\mu(X \cap \sigma (i) )}{\mu(\sigma(i))} .  \] 
Equivalently, $\rho_r$ is the grayscale digital image obtained from the digital approximation to the indicator function for $X$.  

The digital approximation $X(r,t)$ is then the union of closed voxels of size $r$, with $\rho_r$-values at least threshold $t$ with $0<t\leq 1$,
\[  X(r,t) = \bigcup_{\rho_r(i) \geq t} \cl(\sigma(i)). \] 

This definition of digital approximation is a simplified model of segmenting an x-ray CT image. 
If the material consists of just two components with a high degree of x-ray contrast, (e.g., silica and air), then the CT image measures the average x-ray density of each voxel-sized patch of the sample.
In this case, x-ray density encodes the proportion of space occupied by the higher-density material. 
Segmenting converts the grayscale image to a binary image that approximates the distribution in space of the higher-density material.  
The simplest method of segmentation is to use a single threshold value, as we do here.  
In the presence of noise and imaging artefacts this method is unsatisfactory, and there is a considerable literature describing practical methods for achieving good segmentations~\cite{wildenschild_x-ray_2013}. 

Note that changing the choice of threshold $t$ can cause significant changes to the regular homology of $X(r,t)$. 
Fig.~\ref{fig:t_value_blob} shows an example where there is no choice of $t$ for which $X(r,t)$ and $X$ have the same homology. 
This is why we use persistent homology to compare $X$ and $X(r,t)$.  

\begin{figure}
    \centering
    \includegraphics[scale=0.08]{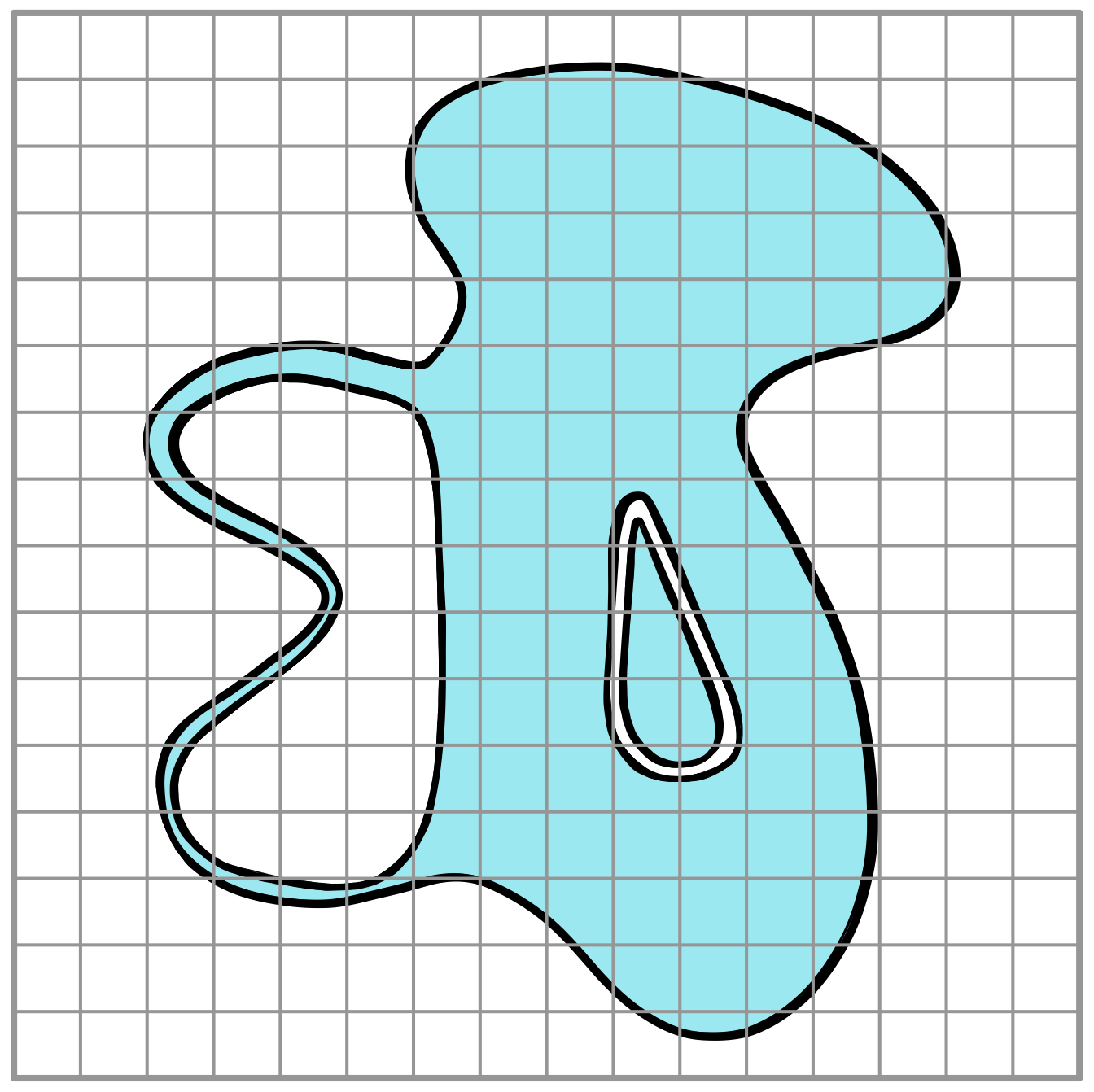}
    \caption{An object $X$ in blue, with a voxel grid of side $r$ overlayed in gray. There is no $t$ that makes $X(r,t)$ have the same homology as $X$. If $t$ is close to $1$, the handle-shaped part of $X$ is lost. If $t$ is small, then the narrow annulus in the middle will be filled in.}
    \label{fig:t_value_blob}
\end{figure}

To quantify the size of the holes and other topological and geometric features of $X$, a suitable function to use is the signed Euclidean distance. 
The \emph{continuous signed Euclidean distance transform} (CSEDT) of $X$ is defined as 
\begin{equation*}
\dt{X}\colon \R^d \to \R; \,
y \mapsto  \begin{cases}
-d(y,\partial X) & \text{if } y \in X\\
d(y,\partial X) & \text{if } y \in X^\compl,
\end{cases}
\end{equation*}
where $\partial X$ denotes the boundary of $X$ and $X^\compl$ denotes the complement of $X$ in $\R^d$.

In practice, since we want to compare $\dt{X}$ and its discretized version, we restrict the distance transforms to a rectangular prism $R$ that contains $X$.
The discrete analog of the CSEDT is the discrete signed Euclidean distance transform (DSEDT) 
\[
 D[X(r,t)] \colon R \to \R
\]
This is a piecewise constant function defined on voxels $\sigma(i) \subseteq X(r,t)$ and $\sigma(j) \subseteq X(r,t)^\compl$ by 
\[  x \mapsto 
\begin{cases}
-\min_{\sigma(j) \subseteq X(r,t)^\compl} \, r\, d(i,j) & \text{if } x \in \sigma(i) \subseteq X(r,t)\\
\min_{\sigma(i) \subseteq X(r,t)} \, r\, d(i,j) & \text{if } x \in \sigma(j) \subseteq X(r,t)^\compl .
\end{cases}
\]
On voxel faces, $D[X(r,t)]$ takes the minimum value over all voxels adjacent to the given face. 
Again, this corresponds to the $T$-construction mentioned earlier. 

In Section~\ref{sec:PH_distance_transform}, we investigate conditions on the geometry of $X$ and the voxel size $r$ necessary to control the distance between the persistence diagrams of the CSEDT of $X$ and the DSEDT of $X(r,t)$.  
Even though the DSEDT depends on the threshold $t$, 
the bounds we obtain do not depend on $t$. 
In other words, even though a wise choice of $t$ might help to achieve a better approximation at a larger voxel size, we can guarantee a good approximation for any value of $t>0$ once $r$ is chosen small enough compared to geometric characteristics of $X$, such as the \emph{reach}. 

Intuitively, the reach of a closed subset $A \subseteq \R^d$ encodes the minimum distance at which two or more fire fronts meet after $A$ is set on fire~\cite{blum}. 
Mathematically, the reach of a closed set $A\subseteq \R^d$ is the largest $\varepsilon$ (possibly $\infty$) such that for every $p\in \R^d$ with distance $d(p,A) < \varepsilon$, $A$ contains a unique point, $\xi_A(p)$, nearest to $p$, i.e. $d(p,A)=d(p,\xi_A(p))$~\cite{federer1959curvature}.
In this paper we use just one property of the reach proven by Federer in \cite[Theorem 4.8 (12)]{federer1959curvature}. For completeness, we give this lemma in Appendix~\ref{sec:app:federer}.

Note that the $\reach(A)$ characterizes the geometry of its complement. Since we work with the signed Euclidean distance transform of a solid object $X$ in this paper, we will use $$\mreach{X} = \min\{\reach(X), \reach(\cl(X^\compl))\}$$ to characterize the geometry of both $X$ and its complement. 
The geometric attributes that determine the reach of $\partial{X}$ are its radii of curvature and distances to the generalised critical points of the signed Euclidean distance transform. 
It is known~\cite{thale_50_2008}, that if $\partial X$ is a closed $C^2$ submanifold of $\R^d$, then $\mreach{X} > 0$. 

In Section~\ref{sec:PH_distance_transform} we establish a bound on the difference between the CSEDT of $X$ and the DSEDT of $X(r,t)$ in terms of a geometric quantity we call the leash. 
We denote by $A_{\delta}$, the set of all points within distance $\delta$ of a set $A$.  
Recall now, that the \emph{erosion} of $X$ by balls of radius $\delta$ is given by 
\[\erosion{X}{\delta} = ((X^\compl)_\delta)^\compl = \{x \in X \,\vert\, B_\delta(x) \subseteq X\}. \]  
The \emph{leash} then measures the Hausdorff distance between $X$ and $E_\delta(X)$: 
\[ \leash{X}{\delta} = \dH(X,\erosion{X}{\delta}) = \sup_{x \in X} d(x,\erosion{X}{\delta}). 
\]
The name leash can be best understood by imagining an infinitesimally small dog connected by a leash to the center of a ball of radius $s$, see Figure~\ref{fig:leash_dog}. 
The ball must stay fully inside the set $A$, while the dog is free to visit the whole of $\R^d$. The minimal length of the leash such that the dog can access every point of $A$ is then $\leash{A}{s}$.

\begin{figure}
    \centering
    \includegraphics[align=c,width=0.9\columnwidth]{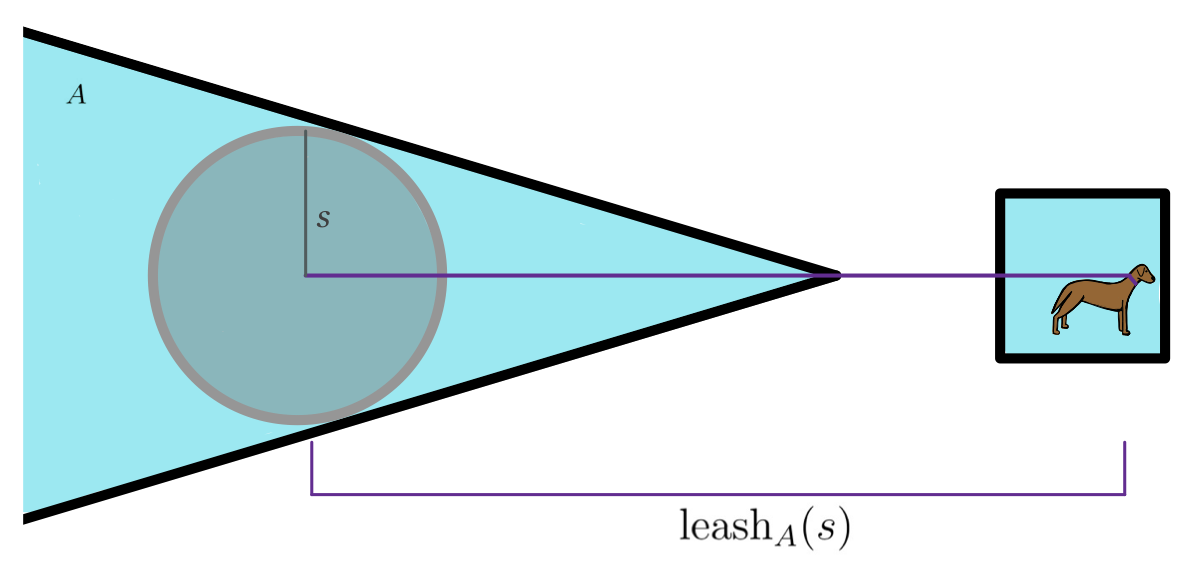}
    \caption{Illustration of the $\leash{A}{s}$: The ball of radius $s$ gets stuck inside part of $A$ but with a long enough leash the dog can reach every part of $A$.}
    \label{fig:leash_dog}
\end{figure}

In Lemma~\ref{lemma:reach_leash}, we show that if the reach is strictly positive and $0< s < \reach(\cl(X)^\compl)$, then $\leash{X}{s} = s$.  
This is equivalent to saying that the closed $\delta$-neighborhood of $E_\delta(X)$ exactly recovers $X$. 

As for the reach, we need to characterize both the geometry of $X$ and its complement $X^c$, so we will use the two-sided leash, defined as 
\[ \mleash{X}{\delta} = \max \{ \leash{X}{\delta}, \leash{\cl(X^\compl)}{\delta} \}.  \]


\section{Persistent Homology of Grayscale Images}\label{sec:PH_grayscale}

Using the definitions and tools outlined in Section~\ref{sec:math_background}, we compare the persistence diagrams of a real-valued function to the persistence diagrams of its digital approximation at a given voxel size.  This is achieved by bounding the difference between the function and its approximation and invoking the stability theorem for persistence diagrams in Proposition~\ref{claim2:gr_bound_higherD}. 
This result has a number of simple corollaries giving us bounds on the differences between persistence diagrams of digital approximations for different voxel sizes, and for the case that the ground-truth function is Lipschitz continuous. 
 


\begin{proposition}\label{claim2:gr_bound_higherD}
	Let $\cf$ be a tame $\mu$-integrable function, and $\dfr$ its digital approximation on a grid with spacing $r$.  Suppose there is a positive number $M_r>0$ such that for every closed voxel $\cl(\sigma(i))$ in the domain of $\cf$, the difference $$\sup_{ x \in \cl(\sigma(i))}  \cf(x) - \inf_{x \in \cl(\sigma(i))} \cf(x) \le M_r.$$
	Then $\dB(PD(\cf), PD(\dfr)) \leq M_r$.
\end{proposition}

\begin{proof}
For any $x$ in the domain of $\cf$, there is at least one voxel $\sigma(i)$ with side-length $r$ such that $x\in \cl(\sigma(i))$.  If $x$ belongs to more than one closed voxel, choose the one for which $\dfr(x) = \dfr(y)$ for $y \in \sigma(i)$.
Let 
	$$M_x := \sup_{y\in \cl(\sigma(i))} \cf(y),\quad \text{and}\quad m_x := \inf_{y\in \cl(\sigma(i))} \cf(y).
	$$
We have $M_r \ge M_x - m_x$.
The definition of $\dfr$ and choice of $\sigma(i)$ ensures that $m_x \leq \dfr(x) \leq M_x$. 	
It follows that 
	\begin{align*}
		|\cf(x) - \dfr(x)| &\leq \max\{\cf(x) - m_x, M_x-\cf(x)\}\\
		& \leq M_x - m_x\\
		& \leq M_r. 
	\end{align*}
We conclude that $\lVert \cf - \dfr \rVert_{\infty} \leq M_r$.
Thus, by the stability theorem \cite{cohen2007stability} the bottleneck distance between the persistence diagrams is also bounded by $M_r$.
\end{proof}
 

 
We now compare the persistence diagrams of the digital approximations $\dfr$ at different resolutions. 
Given two grid spacings $r_1>0$ and $r_2>0$, suppose $r_2$ is divisible by $r_1$, i.e., $a:= r_2/r_1 \in \mathbb{N}$, then any voxel of size $r_2$ contains $a^d$ voxels of size $r_1$. 
Since the measure of the voxel faces is zero, $\cf_{r_2}$ is the $r_2$ digital approximation of $\cf_{r_1}$. 
Therefore, Proposition~\ref{claim2:gr_bound_higherD}, applied to $\cf_{r_1}$, implies an upper bound for the distance between the persistence diagrams of the two digital approximations $\cf_{r_1}$ and $\cf_{r_2}$ as follows.

\begin{corollary}
\label{finer-c1}
Suppose $r_2> r_1 >0$, $a ={r_2}/{r_1} \in \mathbb{N}$, and choose $M\ge 0$ such that for all $d$-voxels $\sigma(i)$ of size $r_2$, the difference 
{\small
\begin{align*}
    \max_{\sigma' \subseteq \sigma(i) } \{ \cf_{r_1} (x) :\; x \in \cl(\sigma') \} 
    - \min_{\sigma' \subseteq \sigma(i) } \{ \cf_{r_1} (x) :\; x \in \cl(\sigma') \}  \leq M ,
\end{align*}}
where $\sigma'$ refers to voxels of size $r_1$.  
Then $\dB(PD(\cf_{r_1}),PD(\cf_{r_2})) \le M$.
\end{corollary}

For the special case when $\cf$ is Lipschitz continuous, we can estimate the constant $M$ in Proposition~\ref{claim2:gr_bound_higherD} by the Lipschitz constant of $\cf$. 
Recall that a function $\cf\colon\R^d \to \R$ is Lipschitz continuous with constant $L >0$ if $|\cf(x)-\cf(y)|\leq L\|x-y\|$.
For brevity, we say $\cf$ is $L$-Lipschitz continuous.  
Note that the CSEDT $\dt{X}$ is Lipschitz continuous with constant $1$.  

\begin{corollary}\label{claim:gr_bound_higherD}
	Suppose $\cf\colon\R^d\to \R$ is $L$-Lipschitz continuous.   
	Then $\dB(PD(\cf), PD(\dfr)) \leq Lr\sqrt{d}$.
\end{corollary}

\begin{proof}
For any voxel $\sigma(i)$, since the function $\cf$ is continuous and $\cl(\sigma(i))$ is compact, there are $x_1,x_2 \in \cl(\sigma(i))$ such that 
$$
	\cf(x_1) = \max_{y\in \cl(\sigma(i))} \cf(y),\quad \text{and}\quad \cf(x_2) = \min_{y\in \cl(\sigma(i))} \cf(y).
$$
It follows that 
\begin{align*}
	\max_{y\in \cl(\sigma(i))} \cf(y) - \min_{y\in \cl(\sigma(i))} \cf(y) & = \lvert \cf(x_1) - \cf(x_2) \rvert \\
	  & \leq L \lVert x_1- x_2 \rVert \leq Lr\sqrt{d}.
\end{align*}
From Proposition~\ref{claim2:gr_bound_higherD}, we see that the bottleneck distance between the persistence diagrams is bounded by $Lr\sqrt{d}$.
\end{proof}


Under the Lipschitz continuity assumption, the bottleneck distance between persistence diagrams of two digital approximations is bounded as follows. 

\begin{corollary}
\label{finer-c2}
Suppose $\cf\colon\R^d \to \R$ is $L$-Lipschitz continuous. Then, $\dB(PD(\cf_{r_1}),PD(\cf_{r_2})) \le Lr_2\sqrt{d}$, when $r_2 > r_1 >0$ and $r_2$ is divisible by $r_1$.
\end{corollary}


When $r_2$ is not an integer multiple of $r_1$, it is still possible to bound the bottleneck distance between the persistent diagrams of $\cf_{r_1}$ and $\cf_{r_2}$ using Proposition~\ref{claim2:gr_bound_higherD} and triangle inequality. However, in this case, the bound is not as tight as the bound in Corollary~\ref{finer-c2}.

\begin{corollary}
\label{finer-c3}
Suppose $\cf\colon\R^d \to \R$ is $L$-Lipschitz continuous. Then, $\dB(PD(\cf_{r_1}),PD(\cf_{r_2})) \le L(r_1+r_2)\sqrt{d}$, for all $r_1, r_2>0$.
\end{corollary}
\begin{proof}
From the proof of Proposition~\ref{claim2:gr_bound_higherD}, triangle inequality and Corollary~\ref{claim:gr_bound_higherD}, we have
\begin{align*}
    || f_{r_1} - f_{r_2}||_{\infty} &\le || f_{r_1} - f||_{\infty} + ||f_{r_2} - f||_{\infty}\\ 
    &\le Lr_1 \sqrt{d} + Lr_2 \sqrt{d} = L(r_1+r_2)\sqrt{d}.
\end{align*} 
By the stability theorem, the bottleneck distance between the persistence diagrams is also bounded by $L(r_1+r_2)\sqrt{d}$.
\end{proof}

\section{Persistent Homology of the Distance Transform of Binary Images}\label{sec:PH_distance_transform}
 
In this section, we consider $X \subseteq \R^d$ as a solid object that is imaged with voxel spacing $r$, yielding its digital approximation $X(r,t)$ for some choice of $0<t\leq 1$.
We derive bounds on the bottleneck distance between the persistence diagram of the DSEDT of $X(r,t)$, and the CSEDT of $X$, by first comparing both of them to the CSEDT of $X(r,t)$, as outlined in Fig.~\ref{fig:Sec4_diagram}.  
\begin{figure}[h!]
    \centering
 \[ \begin{tikzcd}
X\subseteq \R^d  
\arrow{r}{\textrm{discretize}}[swap]{r} \arrow{d}{CSEDT} &\genfrac{}{}{0pt}{1}{\textrm{grayscale}}{\textrm{image}} 
\, \rho_r\arrow{r}{\textrm{threshold}}[swap]{t} &\genfrac{}{}{0pt}{1}{\textrm{binary}}{\textrm{image}}
\, X(r,t) \arrow[swap,gray,near end]{dl}{CSEDT} \arrow[swap]{d}{DSEDT}\\%
\dt{X}\!:\R^d\to\R \arrow{d}{PH} & 
\textcolor{gray}{\dt{\tcx}\!:\R^d\to\R} 
&\genfrac{}{}{0pt}{1}{\textrm{grayscale}}{\textrm{image}} 
\, \tg \arrow[swap]{d}{PH}\\
\PD{\dt{X}} & 
& \PD{\tg} 
\end{tikzcd}
\]
    \caption{Diagram of our model for the digital approximation of a solid object $X$, and its geometric characterisation using signed Euclidean distance transforms and persistent homology. The top row is a simplified version of CT-imaging and segmentation.  
    In gray: In the proof of Thm.~\ref{thm:main}, instead of comparing the continuous distance transform $\dt{X}$ to the discrete distance transform $\tg$ directly, we compare both to the continuous distance transform $\dt{\tcx}$ of the discrete object.
    }
    \label{fig:Sec4_diagram}
\end{figure}
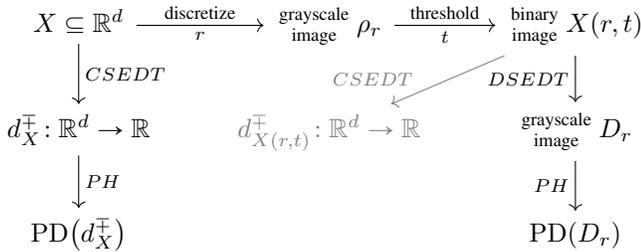

\subsection{Bounds using the reach or leash of $X$} \label{sec:bounds_reach_leash}

The most relevant result of this section is Corollary~\ref{cor:reach_main}, which states that when discretizing with voxel length $r < \frac{1}{\sqrt{d}} \mreach{X}$, the bottleneck distance is bounded by a constant times $r$. 
This follows directly from Theorem~\ref{thm:main}, which bounds the bottleneck distance in terms of the two-sided leash $\mleash{X}{\sqrt{d}r}$ when discretizing with any voxel length $r$, not necessarily smaller than the reach. 
Note that the bounds we give here do not depend on the threshold $t$. 
The results could be slightly improved by incorporating $t$ but we do not do so for the sake of brevity.

\begin{theorem}
\label{thm:main}
Given $X \subseteq \R^d$ with $X=\cl(\inter(X))$.
Let $\dt{X}$ be the CSEDT of $X$, $r>0, \ t\in(0,1]$, and $X(r,t)$ be its digital approximation.
Let $\tg = \DSEDT{X(r,t)}\colon \R^d \to \R$ denote the discrete signed Euclidean distance transform.
Then,
\begin{align*}
\dB(\PD{\dt{X}}, \PD{\tg}) \leq \mleash{X}{\sqrt{d} r} + 2 \sqrt{d} r
\end{align*}
\end{theorem}

The proof requires the following lemmas.
The first lemma shows that the DSEDT of $X(r,t)$ is a good approximation to the CSEDT of $X(r,t)$. 
\begin{lemma}
\label{lemma:DSEDTvsCSEDT}
With the notation of Theorem~\ref{thm:main},
\begin{align*}
\norminf{\dt{\tcx} - \tg } \leq \sqrt{d}r.
\end{align*}
\end{lemma}
\begin{proof}
Let $p \in \R^d$ and assume $p \in \tcx$; the proof for $p \in \tcx^\compl$ works analogously.
Let $y$ be a closest point of $\partial \tcx$ to $p$, i.e.\ $\dt{\tcx}(p) = - d(p, \partial \tcx) = - d(p,y)$, see Fig.~\ref{fig:proof_DSEDT_CSEDT}.
\begin{figure}[hb]
    \centering
    \includegraphics[width = 0.55\columnwidth]{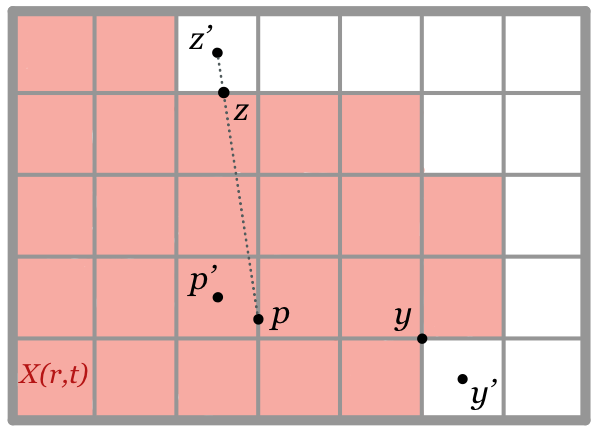}
    \caption{Illustration of the variables appearing in the proof of Lemma~\ref{lemma:DSEDTvsCSEDT}.}
    \label{fig:proof_DSEDT_CSEDT}
\end{figure}
As $y \in \partial \tcx$, there exists a voxel $\sigma(i) \subseteq \tcx^\compl$ with voxel center $y^\prime$ of distance $d(y,y^\prime) \leq \frac{1}{2} \sqrt{d} r$. 
Let $\sigma(j)$ be the voxel with $p$ in its closure with minimal filtration value (in case there are several), i.e.\ the voxel giving filtration value $\tg(p)=\tg(\sigma(j))$ to $p$, and let $p^\prime$ be its center. 
Thus, $d(p,p^\prime)\leq \frac{1}{2} \sqrt{d} r$. Hence,
\begin{align*}
    -\tg(p) &= -\tg(\sigma(j)) \\
     &= \min\{ d(p^\prime,c) \ | \ c \textrm{ voxel center of } \sigma(k) \subseteq \tcx^\compl \} \\
     &\leq d(p^\prime,y^\prime) \leq d(p^\prime,p) + d(p,y) + d(y,y^\prime) \\
     &\leq \frac{1}{2}\sqrt{d}r - \dt{\tcx}(p) + \frac{1}{2}\sqrt{d}r,
\end{align*}
yielding $\dt{\tcx}(p)-\tg(p) \leq \sqrt{d}r$.

Let $z^\prime$ be the voxel center minimizing 
$$\min\{ d(p^\prime,c) \ | \ c \textrm{ voxel center of } \sigma(k) \subseteq \tcx^\compl \} = -\tg(p).$$ 
As $z^\prime \in \tcx^\compl$ and $p \in \tcx$, there exists a $z \in \partial \tcx$ on the straight line segment between $z^\prime$ and $p$. Hence,
\begin{align*}
    - \dt{\tcx}(p) &= d(p,\partial \tcx) \leq d(p,z) \leq d(p,z^\prime) \\
    &\leq d(p,p^\prime) + d(p^\prime,z^\prime) \leq \frac{1}{2}\sqrt{d}r - \tg(p),
\end{align*}
yielding $-(\dt{\tcx}(p)-\tg(p)) \leq \frac{1}{2}\sqrt{d}r$.
\end{proof}

The authors of \cite{cohen2007stability} mention that, by definition,
the Hausdorff distance between two sets equals the $L_\infty$-distance between the (unsigned) Euclidean distance transforms of these two sets. 
For the proof of Theorem~\ref{thm:main} an extension of this to the \emph{signed} Euclidean distance transform is needed, see Lemma~\ref{lemma:generalizeCS-E-H}. Note that the bound is between the maximum and the sum of the two terms $\dH(A,B)$ and $\dH(A^\compl,B^\compl)$.
\begin{lemma}
\label{lemma:generalizeCS-E-H}
Let $A,B \subseteq \R^d$, then
\begin{align*}
\norminf{\dt{A}-\dt{B}} &\leq \max \{ \sup_{b\in B^\compl} d(A^\compl,b) + \sup_{a\in A} d(a,B), \\
& \quad \quad \quad \quad \quad \sup_{b\in B} d(A,b) + \sup_{a\in A^\compl} d(a,B^\compl) \}
\end{align*}
\end{lemma}
\begin{proof}
Let $d_A$ denote the unsigned Euclidean distance transform defined as $d_A(p) = d(p,A)$, for any $p \in \R^d$. 
As mentioned above, we have by definition that $\norminf{\dtonesided{A}-\dtonesided{B}} = \dH(A,B)$. 
To bound $\norminf{\dt{A}-\dt{B}}$, we consider four cases for the point $p\in\R^d$.\par  
   \textit{Case 1:} $p \in A , \ p \in B$.
   Then $\absolute{\dt{A}(p)-\dt{B}(p)} = \absolute{-\dtonesided{A^\compl}(p)+\dtonesided{B^\compl} (p)} \leq \norminf{\dtonesided{B^\compl}-\dtonesided{A^\compl}} = \dH(A^\compl, B^\compl) = \max \{ \sup_{b\in B^\compl} d(A^\compl,b), \ \sup_{a\in A^\compl} d(a,B^\compl) \}$.\par 
   \textit{Case 2:} $p \notin A ,\  p \notin B$.
   Analogously $\absolute{\dt{A}(p)-\dt{B}(p)} \leq \max \{ \sup_{a\in A} d(a,B), \ \sup_{b\in B} d(A,b) \} $.\par 
   \textit{Case 3:} $p \in A ,\  p \notin B$. 
   Then $\absolute{\dt{A}(p)-\dt{B}(p)} =
   \absolute{-d(A^\compl,p) - d(p,B)}
   = d(A^\compl,p) + d(p,B) \leq \sup_{b\in B^\compl} d(A^\compl,b) + \sup_{a\in A} d(a,B)$.\par 
   \textit{Case 4:} $p \notin A ,\  p \in B$.
   Analogously $\absolute{\dt{A}(p)-\dt{B}(p)} \leq \sup_{b\in B} d(A,b) + \sup_{a\in A^\compl} d(a,B^\compl)$. 
\end{proof}
\begin{remark} \label{remark:tighteningCS-E-H}
In the special case that $\max \{ \dH(A,B),\dH(A^\compl,B^\compl) \} < \mreach{A}$, Lemma~\ref{lemma:generalizeCS-E-H} can be tightened to $\norminf{\dt{A}-\dt{B}} \leq \max \{ \dH(A,B),\dH(A^\compl,B^\compl) \}$, 
see Appendix~\ref{sec:app:tightening}.
\end{remark}

\begin{figure}
    \centering
    \includegraphics[align=c, width=0.9\columnwidth]{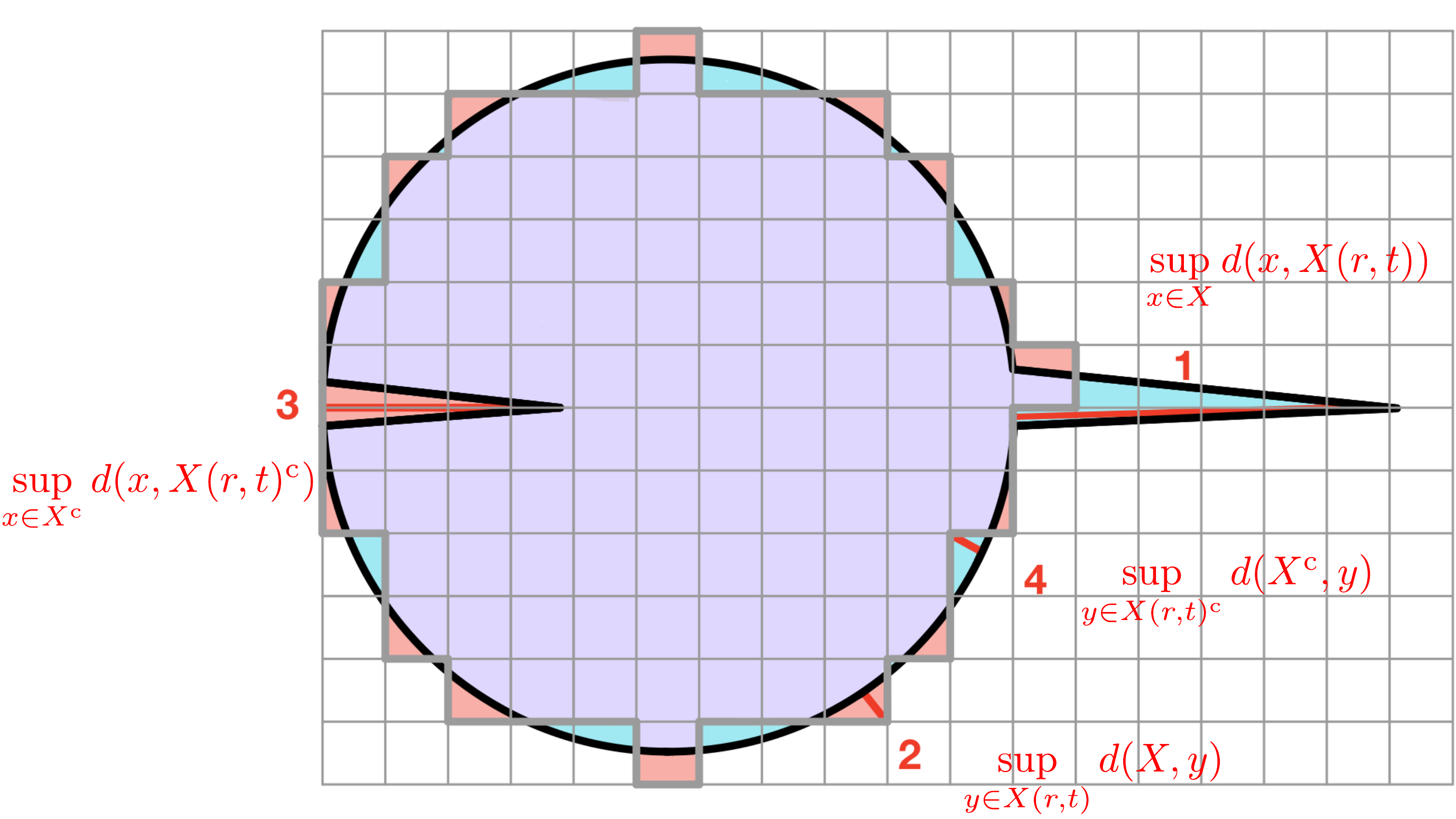}
    \caption{The four different suprema in the bound of Lemma~\ref{lemma:generalizeCS-E-H} can all have different values. In the proof of Theorem~\ref{thm:main} we show that suprema 1 and 3 are bounded by the two-sided leash $\mleash{X}{\sqrt{d}r}$, while suprema 2 and 4 are bounded by the voxel diameter $\sqrt{d}r$.}
    \label{fig:4suprema}
\end{figure}

However, when the sets differ by more than their reach, the bound in Lemma~\ref{lemma:generalizeCS-E-H} can be tight, as 
Example~\ref{example:tight_bound} shows.
\begin{example}\label{example:tight_bound}
Let $A=[-1.5,1.5] \times \R$ and $B=[-1.5,1.5] \times \left( (-\infty,-1] \cup [1,\infty)  \right)$. 
Then, the bound given in Lemma~\ref{lemma:generalizeCS-E-H} is tight: $\norminf{\dt{A}-\dt{B}} = \absolute{\dt{A}(0)-\dt{B}(0)} = \absolute{-1.5-1}= 1.5 + 1= \sup_{b\in B^\compl} d(A^\compl,b) + \sup_{a\in A} d(a,B)$. 
\end{example}


\begin{lemma} \label{lemma:bound_with_suprema}
With the notation of Theorem~\ref{thm:main},
\begin{align*}
&\dB(\PD{\dt{X}}, \PD{\tg})\\
&\leq \max \{ \sup_{y\in \tcx^\compl} d(X^\compl,y) + \sup_{x\in X} d(x,\tcx) , \\
&\quad \quad \quad \ \, \sup_{y\in \tcx} d(X,y) + \sup_{x\in X^\compl} d(x,\tcx^\compl) \} + \sqrt{d} r
\end{align*}
\end{lemma}
\begin{proof}
The proof is a combination of the stability theorem of persistent homology \cite{cohen2007stability}, the triangle inequality, and Lemmas~\ref{lemma:generalizeCS-E-H}, \ref{lemma:DSEDTvsCSEDT}:
\begin{align*}
&\dB(\PD{\dt{X}}, \PD{\tg})\\
&\leq \norminf{\dt{X} - \tg} \\
&\leq \norminf{\dt{X} - \dt{\tcx}} + \norminf{ \dt{\tcx} - \tg} \\
&\leq \max \{ \sup_{y\in \tcx^\compl} d(X^\compl,y) + \sup_{x\in X} d(x,\tcx) , \\
&\quad \quad \quad \ \, \sup_{y\in \tcx} d(X,y) + \sup_{x\in X^\compl} d(x,\tcx^\compl) \} + \sqrt{d} r.
\end{align*}
\end{proof}

\begin{lemma} \label{lemma:bounding_suprema_by_leash}
With the notation of Theorem~\ref{thm:main},
\begin{enumerate}
    \item $\sup_{x\in X} d(x,\tcx) \leq \leash{X}{\sqrt{d}r}$
    \item $\sup_{x\in X^\compl} d(x,\tcx^\compl) \leq \leash{\cl(X^\compl)}{\sqrt{d}r}$
    \item $\sup_{y\in \tcx} d(X,y) \leq
    \sqrt{d}r$
    \item $\sup_{y\in \tcx^\compl} d(X^\compl,y) \leq
    \sqrt{d}r$
\end{enumerate}
\end{lemma}
\begin{proof}
1) To prove $\sup_{x\in X} d(x,\tcx) \leq \leash{X}{\sqrt{d}r}$, let $x\in X$ be arbitrary. 
By the definition of $\leash{X}{\sqrt{d}r}$, there exists a point $a \in ((X^\compl)_{\sqrt{d}r})^\compl$ with $d(x,a)\leq \leash{X}{\sqrt{d}r}$. 
As $a$ is at least $\sqrt{d}r$ far from $X^\compl$, the open ball $\ballop{\sqrt{d}r}{a}$ is fully inside $X$. 
Let $\sigma$ be a voxel containing $a$ in its closure. 
As the diameter of the voxel is $\diam(\sigma)=\sqrt{d}r$, the voxel $\sigma \subseteq \ballop{\sqrt{d}r}{a}\subseteq X$ has density $\rho_r( \sigma )=1\geq t$, for all $t \in (0,1]$. 
Hence, $a \in \cl(\sigma) \subseteq \tcx$ and $d(x,\tcx) \leq d(x,a) \leq \leash{X}{\sqrt{d}r}$.

2) The proof for $\sup_{x\in X^\compl} d(x,\tcx^\compl) \leq \leash{\cl(X^\compl)}{\sqrt{d}r}$ is analogous.

3) To prove $\sup_{y\in \tcx} d(X,y) \leq \sqrt{d}r$,
let $y\in \tcx$ be arbitrary.
Let $\sigma \subseteq \tcx$ be such that $y \in \cl (\sigma)$.
As $\rho_r(\sigma)>t>0$, there exists a point $x \in X \cap \sigma$.
Hence, $d(X,y) \leq d(x,y) \leq \diam(\sigma) = \sqrt{d}r$.
    
4) The proof for $\sup_{y\in \tcx^\compl} d(X^\compl,y) \leq \sqrt{d}r$
is analogous.
\end{proof}

\begin{proof}[Proof of Theorem~\ref{thm:main}]
Lemma~\ref{lemma:bounding_suprema_by_leash} shows that both the terms $\sup_{y\in \tcx^\compl} d(X^\compl,y) + \sup_{x\in x} d(x,\tcx)$ and $\sup_{y\in \tcx} d(X,y) + \sup_{x\in X^\compl} d(x,\tcx^\compl)$ are bounded by $\sqrt{d}r + \mleash{X}{\sqrt{d}r}$. In combination with Lemma~\ref{lemma:bound_with_suprema} this finishes the proof.
\end{proof}

In the special case when the voxel diameter $\sqrt{d}r$ is smaller than $\mreach{X}$, we can calculate the term $\mleash{X}{\sqrt{d}r}$ explicitly using the following Lemma.
\begin{lemma}\label{lemma:reach_leash}
Let $A$ be a subset of $\R^d$ with $A=\cl(\inter(A))$ and
$\reach(\cl(A^\compl))>0$.
Let $s<\reach(\cl(A^\compl))$, 
then $\leash{A}{s} = s$.
\end{lemma}
\begin{proof}
To show that $\leash{A}{s} = \sup_{a \in A} d(a,\erosion{A}{s})$ is at most $s$, we need to find for every $a \in A$ an element $b \in \erosion{A}{s}$ with distance $d(a,b)$ at most $s$.
Let $a \in A$ be arbitrary. If $a \in \erosion{A}{s}$, then $b=a$ fulfills $d(a,b)=0 \leq s$.
Therefore $a \notin \erosion{A}{s}$, i.e. $a \in (A^\compl)_s$. 
Let $x \in \partial \cl(A^\compl) $ 
be s.t. $d(a,x)=d(a, \cl(A^\compl))$.
If $x\neq a$, then $n_x = \frac{a-x}{\norm{a-x}}$ is a unit normal vector
of $\cl(A^\compl)$ at $x$ (for a rigorous definition of normal vector, see Definition~\ref{def:normal}). 
If $x=a$, let $n_x$ be any unit normal vector of $\cl(A^\compl)$ at $x$. 
Define $b = x + s n_x$. 
Since $s < \reach(\cl(A^\compl))$, Federer's Lemma~\ref{lemma:federer_rigorous} yields, $d(A^\compl,b)= d(\cl(A^\compl),b) = d(x,b)=s$. 
Hence, $b\notin (A^\compl)_s$ and thus $b \in ((A^\compl)_s)^\compl = \erosion{A}{s}$. 
The points $x$, $a$, and $b$ lie on a straight line by construction.
The point $a$ has distance $d(a,x)=d(a, \cl(A^\compl)))=d(a,A^\compl)$ to $x$, which is less than $s$ as $a \in (A^\compl)_s$. 
Hence, $a$ lies on the line segment between $x$ and $b$, and thus $d(a,b)\leq d(x,b)=s$, concluding the proof of $\leash{A}{s} \leq s$.

For any $a \in \partial A$, we get $d(a,\erosion{A}{s}) \geq s$, yielding $\leash{A}{s} = s$.
\end{proof}

\begin{corollary}
\label{cor:reach_main}
Using the notation of Theorem~\ref{thm:main}.
If $r < \frac{1}{\sqrt{d}} \mreach{X}$, then 
\begin{align*}
\dB(\PD{\dt{X}}, \PD{\tg}) \leq 3 \sqrt{d} r
\end{align*}
\end{corollary}
\begin{proof}
Apply Theorem~\ref{thm:main} and Lemma~\ref{lemma:reach_leash}.
\end{proof}
\begin{remark} \label{remark:tightening_cor}
Using Remark~\ref{remark:tighteningCS-E-H}, the constant $3$ in Corollary~\ref{cor:reach_main} can be tightened to $2$, 
see Appendix~\ref{sec:app:tightening}.
\end{remark}

\subsection{Bounds using the density $\rho_r$} \label{sec:bounds_rho}

Suppose we start with $\rho_r$, the function that tells us the proportion of $X$ in each voxel defined in Section~\ref{subsec:images_object}. 
As mentioned there, this is a reasonable model of the information contained in an x-ray CT image of a two-phase material with high x-ray contrast. 
Theorem~\ref{thm:Vanessas_bound} bounds the bottleneck distance between the persistence diagrams of $\dt{X}$ and $\tg$ in terms of $\rho_r$. 
The proof reuses Lemma~\ref{lemma:bound_with_suprema}, but we need to bound the suprema by quantities other than the leash this time.

\begin{lemma} \label{lemma:bounding_suprema_by_m}
Using the notation of Theorem~\ref{thm:main}. 
\begin{enumerate}
    \item $\sup_{x\in X} d(x,\tcx) \leq \max_{\rho_r(i)>0} \tg(\sigma(i))$
    \item $\sup_{x\in X^\compl} d(x,\tcx^\compl) \leq - \min_{\rho_r(i)<1} \tg(\sigma(i))$.
\end{enumerate}
\end{lemma}
\begin{proof}
We prove the first statement here; the second can be shown similarly. 
Let $x \in X$, and assume $x \notin \tcx$ (otherwise the statement follows trivially). 
As $x\in X=\cl(\inter(X))$, any voxel $\sigma(j)$ containing $x$ in its closure has $\rho_r(j)>0$. Let $x^\prime$ be the center of $\sigma(j)$. As $x \notin \tcx$, its voxel $\sigma(j)$ is in $\tcx^\compl$. Let $y^\prime$ be the voxel center minimizing 
    $$\min\{ d(x^\prime,c) \ | \ c \textrm{ voxel center } \in \tcx \} = \tg(\sigma(j)).$$ 
    As $x-x^\prime$ is a vector pointing from a voxel center to a point in the closure of the same voxel, adding this vector to the voxel center $y^\prime$ yields a point in the closure of the same voxel, and thus a point in $\tcx$. Hence,
    \begin{align*}
    d(x,\tcx) &\leq d(x,y^\prime + (x-x^\prime)) \\
    &= d(x - (x-x^\prime) ,y^\prime) \\
    &= d(x^\prime,y^\prime) \\
    &= \min\{ d(x^\prime,c) \ | \ c \textrm{ voxel center } \in \tcx \} \\
    &= \tg(\sigma(j)) \leq \max_{\rho_r(i)>0} \tg(\sigma(i)).
    \end{align*}
\end{proof}

Define  
$\mm$ as the maximum of $\max_{\rho_r(i)>0} \tg(\sigma(i))$ and $- \min_{\rho_r(i)<1} \tg(\sigma(i))$. 
This measures how far $X(r,t)$ is from the two most extreme ways of thresholding.  

\begin{theorem} \label{thm:Vanessas_bound}
With the notation of Theorem~\ref{thm:main},
\begin{align*}
\dB(\PD{\dt{X}}, \PD{\tg}) \leq \mm + 2 \sqrt{d} r
\end{align*}
\end{theorem}
\begin{proof}
Combining Lemma~\ref{lemma:bound_with_suprema}, Lemma~\ref{lemma:bounding_suprema_by_m}, and items 3 and 4 of Lemma~\ref{lemma:bounding_suprema_by_leash}.
\end{proof}

\subsection{Examples}

We describe an example that illustrates many of the bounds derived above.  
The ground-truth object $X$ consists of an array of small circular spots lying inside a large disk together with the outside of this disk; see the image at top left of Fig.~\ref{fig:small_dot}.
In total this image is 2048$\times$2048 pixels and the large white disk is a circle of radius $R_2=510$ pixels.
The small black disks have radius $R_1=5$ pixels and  centers $w=85$ 
pixels apart. 

The reach of $\partial X$ is the radius $R_1$ of the small spots, so the image with $2048^2$ pixels has $1 = r <   \mreach{X}/\sqrt{2} = 3.54$.  
For any resolution with $r<3.54$, (no.~pixels $>580^2$)  Remark~\ref{remark:tightening_cor} guarantees that the bottleneck distance between the PDs for $\dt{X}$ and $\DSEDT{X(r,t)}$ is no larger than $2\sqrt{2} r$. 

At coarser resolutions, Theorem~\ref{thm:main} still applies with $\mleash{X}{\sqrt{d} r} = R_2 - \frac{w}{\sqrt{2}}+R_1+\sqrt{2}r = 454.9 +\sqrt{2}r$ when
$\mreach{X} \leq \sqrt{2} r \leq \frac{w}{\sqrt{2}}-R_1$. 
So we see that
\[ \dB(\PD{\dt{X}}, \PD{\DSEDT{X(r,t)}}) \leq 454.9 +3 \sqrt{2} r, \]
for $3.55 \leq r \leq 39.0$. 

For this synthetic example, we can also compute $\rho_r$ for all choices of $r$.  
Let $\epsilon>0$ be a small tolerance to accommodate noise.  
The pixels where $\epsilon < \rho_r < 1-\epsilon$ are those that have non-empty intersection with $\partial X$. 
At high resolutions, i.e., pixel sizes $r < \mreach{X}/\sqrt{2}$, this means the Hausdorff distances between $X(r,t)$ and $X(r,1-\epsilon)$ or $X(r,\epsilon)$ are small.  

At coarser resolutions, for example the image with $64\times 64$ pixels depicted at lower left of Fig.~\ref{fig:small_dot}, each spot is smaller than a pixel and all pixels intersecting the large disk have $\epsilon < \rho_r < 1-\epsilon$. 
In this situation, setting $t=1-\epsilon$ means the set $X(r,1-\epsilon)$ is contained in the outside of the large disk, while $X(r,\epsilon)$ is the entire square.  

Fig.~\ref{fig:small_dot} shows that the red bound from Section~\ref{sec:bounds_rho} has the advantage of being tighter than the green bound from Section~\ref{sec:bounds_reach_leash} at lower resolutions (i.e., pixel sizes larger than the reach). 
However, at high resolutions the bound given by Remark~\ref{remark:tightening_cor} is better. 

 
\begin{figure}
    \centering
    \includegraphics[align=c,width=0.48\columnwidth]{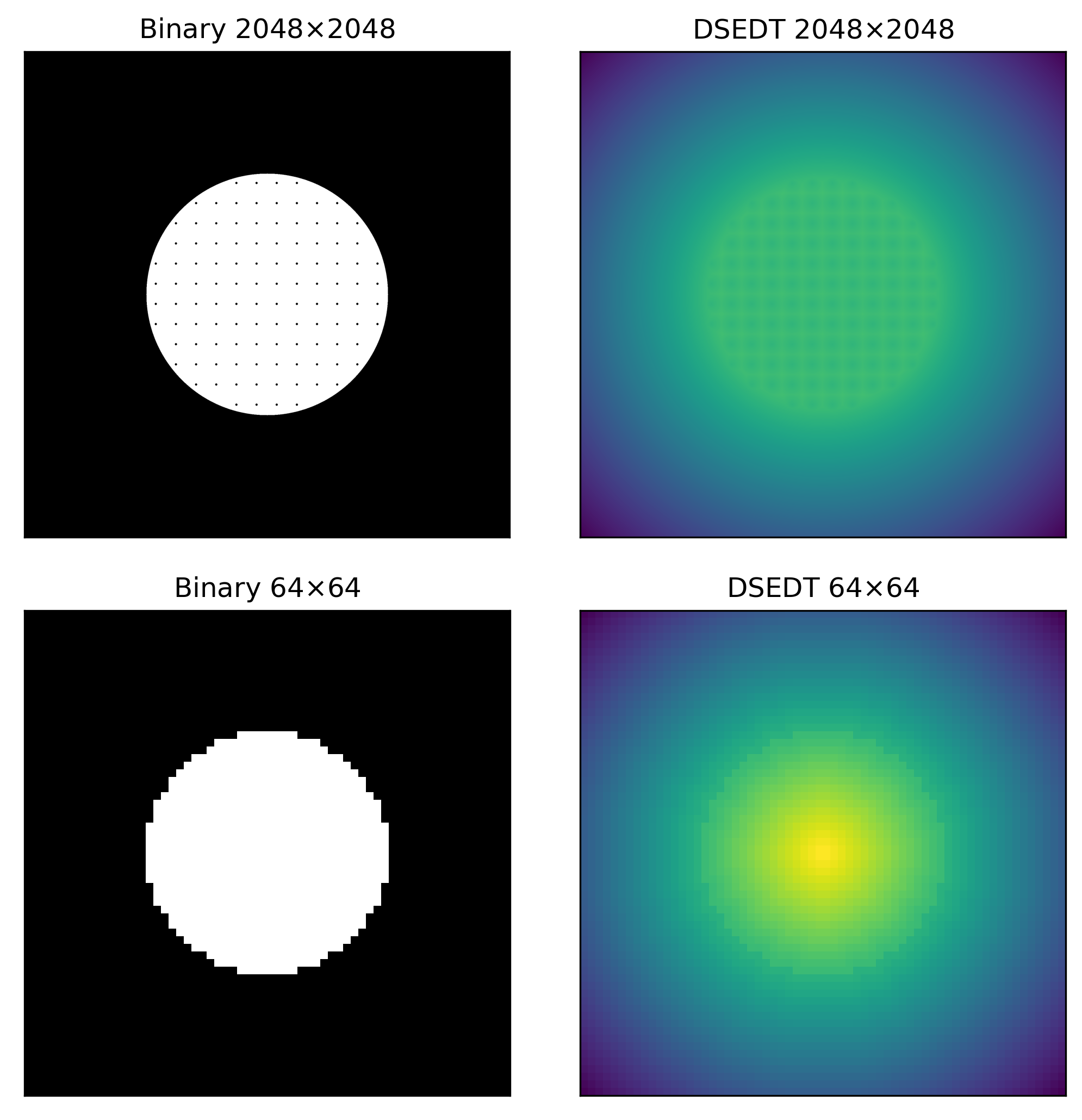}
    \includegraphics[align=c,width=0.5\columnwidth]{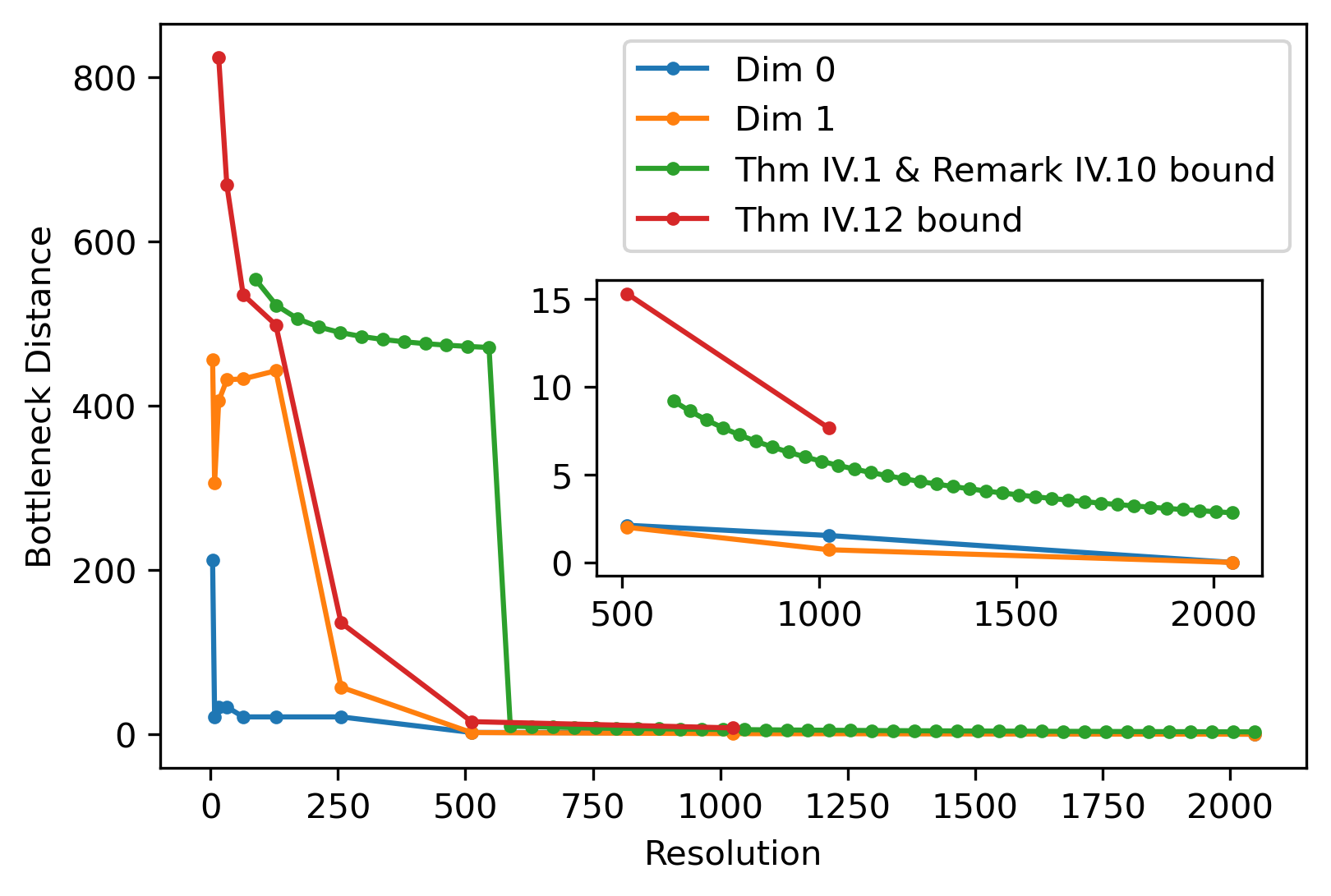}    \\
    \caption{ 
    Left two columns: The top row shows a high resolution binary image with corresponding DSEDT and the bottom row shows a lower resolution version obtained by downsampling with the averaging method and thresholding with $t=0.5$ and its corresponding DSEDT. Right column: Bottleneck distances comparing persistence diagrams for each resolution with that for the highest resolution. The distances for dimensions 0 and 1 are shown along with the bounds provided by Theorem~\ref{thm:main} with Remark~\ref{remark:tightening_cor} (green) and Theorem~\ref{thm:Vanessas_bound} (red).}  
    \label{fig:small_dot}
\end{figure}

 
\section{Applications}
\label{sec:applications}

In any real-world situation, the reach and the leash for the object being imaged are likely to be unknown, and/or computationally expensive to estimate. 
Similarly, the density $\rho_r$ at a given pixel length $r$ might be known, but noisy.
So in this section, we explore a different approach to determining an adequate digital resolution using the persistence diagrams of DSEDTs computed for a succession of larger voxel sizes, i.e., decreasing image dimensions. 
 
In our case studies, we begin with the highest resolution image and downsample to lower resolution images as follows. 
For a $d$-dimensional binary image with grid spacing $r$, and dimensions $n_1 \times \cdots \times n_d$, we compute the averages in blocks of $a^d$ voxels, where $a$ is called the kernel size.
Based on these average values, we create the downsampled binary image with grid spacing $r' = a r$ and dimensions $n_1/a \times \cdots \times n_d/a$ by thresholding at value $t = 0.5$, similar to the process described in Section~\ref{subsec:images_object}.
From the new image, we compute the DSEDT, with the new grid spacing $r'$ used as a scaling factor.

All examples in this section are square, so the image size is $n^d$ and we use $n$ to quantify resolution.  
The voxel size is specified with arbitrary units so that $r=1$ for the highest-resolution image in each case. 
To avoid issues caused by boundary effects, we only choose kernel sizes that fit evenly within the original high-resolution image dimensions.
Thus kernel sizes are always integer divisors of the image size.
This gives us a collection of images of different resolutions that are approximations of the highest resolution image.
Then we compute the bottleneck distance between the persistence diagram for each resolution with the persistence diagram from the high resolution image.
This allows us to quantify how much information is lost when we downsample. 

We compute the persistence diagrams and bottleneck distance, using the Giotto-TDA python package\cite{tauzin2020giottotda}. 
Code for this project, including code to generate the synthetic examples, can be found on Github\cite{phimages_githubrepo}.


\subsection{Structure at Different Length Scales}
\label{sec:theory_not_apply}

Our goal in this section is to explore how the persistent homology of a DSEDT changes at different image resolutions when approximating a particular object $X \subset \R^d$, with structure at different length scales, illustrated in Fig.~\ref{fig:nested_rings}. 
Although the overall trend is that bottleneck distance decreases with resolution, this decrease is not monotonic.  
The distinct length scales mean the distances show a succession of plateaus as each feature is resolved and the persistence diagrams remain relatively stable over an interval of resolutions.

We now define a $\Delta$-$\varepsilon$ plateau to capture the change in bottleneck distance between persistence diagrams as the image resolution changes. 
Recall that the (linear) image size $n$ is used to quantify  resolution, and, for the remainder of this section, denote the DSEDT $\DSEDT{X(r,t)}$ by $D^n$. 
Let $N$ denote the highest resolution available and take an interval $\Delta = [\ell,m]$ of image resolutions and $\varepsilon > 0$. We then say there is a \emph{$\Delta$-$\varepsilon$ plateau} if for all $j,k \in [\ell,m]$, 
\[\absolute{\dB\left(\PD{D^j}, \PD{D^N}\right)  - \dB\left(\PD{D^k}, \PD{D^N}\right)} < \varepsilon. \]
By the triangle inequality $\dB \left(\PD{D^j}, \PD{D^k}\right)<\varepsilon$ 
is a sufficient condition for a $\Delta$-$\varepsilon$ plateau.

As an example, consider the image in Figure~\ref{fig:nested_rings} with three nested rings of different thicknesses. 
The original high-resolution image, $X_N$, has $r=1$ and $N^2 = 5040^2$ pixels. 
We downsample the original image to create 57 different images, where each new image has a kernel of size $a$ where $a$ is a proper divisor of $5040$, so that the resolution in each case is $n = 5040/a$, and the voxel size at this resolution is $r=a$.
The DSEDT and persistence diagrams are computed for each image, and then we find the bottleneck distance $\dB(\PD{D^n},\PD{D^{5040}})$ comparing each lower-resolution image with the original one.
The plot of these results in Figure~\ref{fig:nested_rings}
shows three specific behaviors we would like to emphasize; spikes, plateaus, and final plateaus. 

As we increase the resolution, some fluctuations in the bottleneck distance are to be expected, even within a plateau. These fluctuations are on the order of magnitude of the pixel diameter.
For the example in Fig.~\ref{fig:nested_rings} it makes a difference whether the center of the image is the center of a pixel (for $n$ odd) or the center is in the closure of $4$ pixels 
(for $n$ even). 
These effects can
cause the bottleneck distance to spike, i.e.\ to increase and then immediately decrease.
The lower the resolution, the more pronounced these spikes usually appear.
The blue zero-dimensional bottleneck distance curve in Fig.~\ref{fig:nested_rings} has a spike at $n=21$, the only odd number in this range of resolutions, with an increase by $366$ from $n=20$ to $n=21$ and a decrease by $333$ from $n=21$ to $n=24$. The magnitude of this spike should be compared with the pixel diameter for $n=21$, which is $r\sqrt{2} = 240 \sqrt{2} = 339.4$.

Next, we observe the plateau behavior outlined in the definition of a $\Delta$-$\varepsilon$ plateau.
This is caused by the introduction of a new topological feature, such as a grain of sand or a small pore.
In the nested ring example, the plot of one-dimensional bottleneck distances shows three plateaus. 
The zero-dimensional bottleneck distances mimics this behavior, but the first one or two plateaus are overshadowed by noise, like spikes, explained above.
The three plateaus are due to the three different ring widths. 
After each of these is resolved, an increase in resolution does little to change the persistence diagram. 
As shown in Fig.~\ref{fig:nested_rings}, we see just one ring at resolution $18$, two rings at $n=40$, 
the third ring starts being resolved at $n=105$ (where it only consist of 4 pixels) and gets fully resolved at $n=210$. 
Specifically, the first $\Delta$-$\varepsilon$ plateau is at $\Delta=[18,24]$ with $\varepsilon=36.96$ in dimension 1 and $\varepsilon=402.51$ in dimension 0.
The second $\Delta$-$\varepsilon$ plateau is at $\Delta=[40,90]$ with $\varepsilon=17.31$ in dimension 1 and $\varepsilon=104.49$ in dimension 0.
The third $\Delta$-$\varepsilon$ plateau is at $\Delta=[210,5040]$ with $\varepsilon=13.43$ in dimension 1 and $\varepsilon=15.80$ in dimension 0.

The existence of structure at just a few distinct, well separated length scales in the synthetic example means there is a sequence of plateaus as each structure is resolved. 
This does not happen for the porous materials discussed in Section \ref{sec:matsci}.
However, we note that Corollary~\ref{cor:reach_main}, and Remark~\ref{remark:tightening_cor} guarantee the existence of a final plateau for any object with positive reach.
Specifically, let $M$ be the lowest resolution required for the corresponding voxel size $r_M < \tfrac{1}{\sqrt{d}} \mreach{X}$, and assume that the highest resolution image has $N>M$.
Then the bottleneck distances will have a $\Delta$-$\varepsilon$ plateau with $\Delta = [M,N]$ and  
$\varepsilon = 4 \sqrt{d} \, r_M \leq 4\mreach{X}$. 
In practice, we often see a final plateau even when the reach is zero, such as 
when the leash bounding the bottleneck distance converges to zero.
The glass bead packing in the following section is such an example. 
 
Although a final plateau suggests the image resolution is sufficient to capture the actual underlying structure of the imaged object, we can never definitively know whether a plateau is final or not, because we do not know what structure exists at length-scales finer than the voxel size.
Hence, the plateau serves only as a guide for resolution choice. 

 

\begin{figure}
    \centering
    \includegraphics[width=0.7\columnwidth]{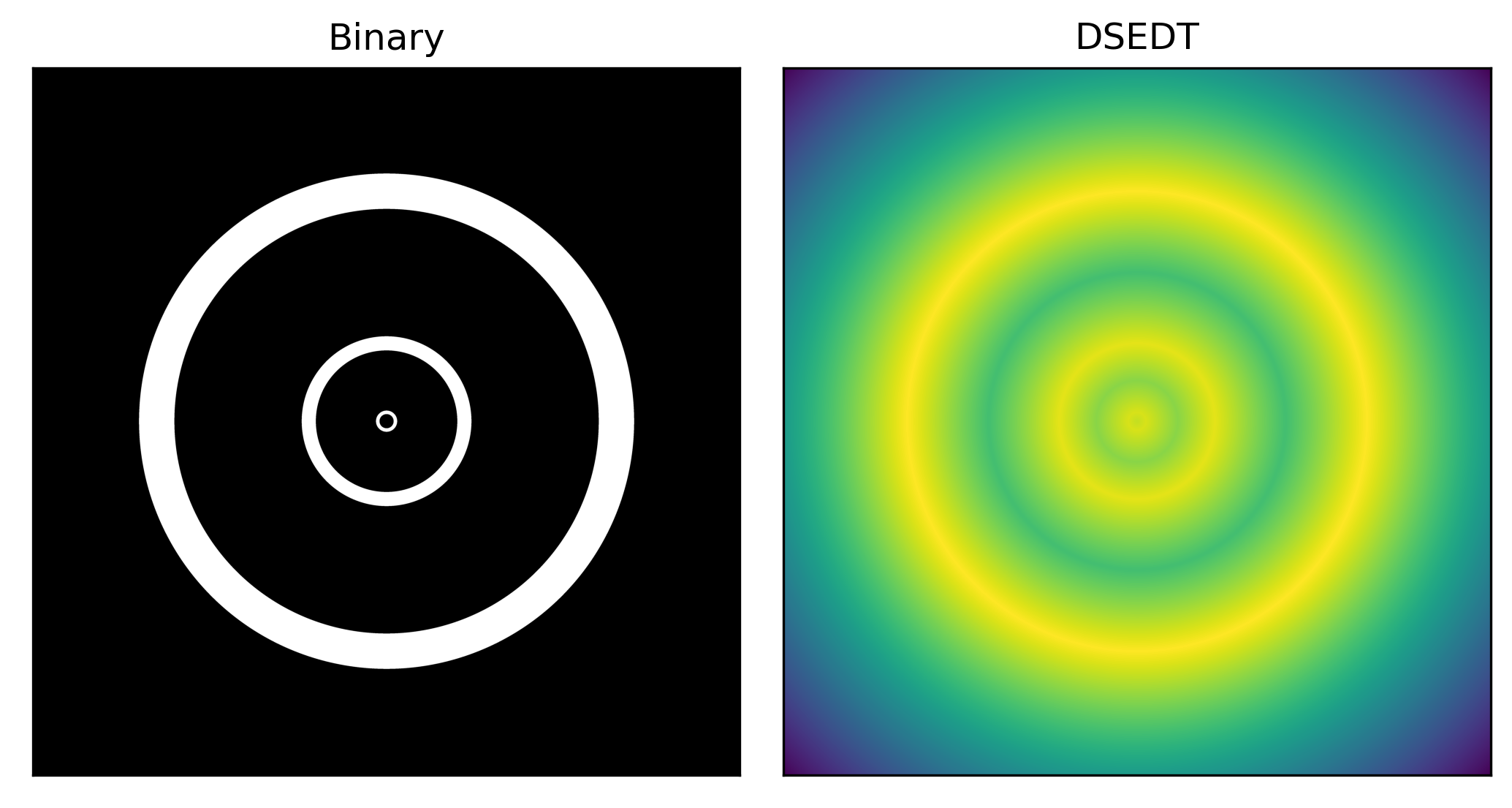}\\
    \includegraphics[width=0.8\columnwidth]{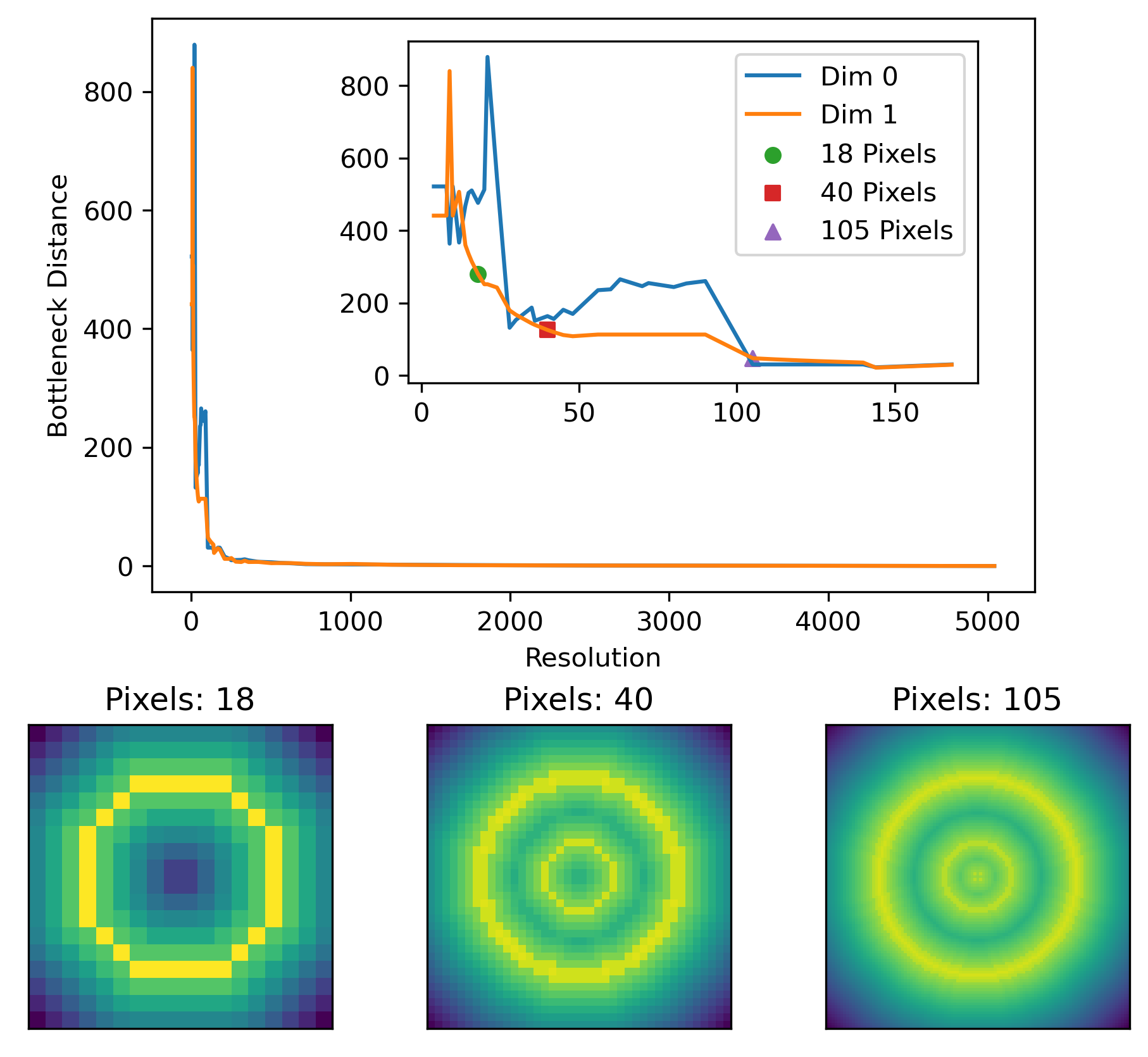}
    \caption{Top, the digital image we use as the ground truth with the corresponding DSEDT. Middle, the plot of the bottleneck distances between 1-dimensional persistence diagrams from each resolution and the highest resolution. Below, the DSEDT of the downsampled image at three different resolutions.}
    \label{fig:nested_rings}
\end{figure}

\subsection{Application to Porous Materials}
\label{sec:matsci}

It is common for material science examples to have zero reach. 
For example, packings of spherical or elliptical grains have $\reach{(X)}=0$, because two grains touch at a single point, while other porous materials, such as a metal foam, have sharp corners. 
Here we present three examples of micro-CT images of porous materials segmented into the two phases of solid and void~\cite{sheppard_network_2005}.
The images are subregions from a packing of spherical glass beads, a sandstone (from Castlegate) and an unconsolidated sandpack, each with $512^3$ voxels. 

In Fig.~\ref{fig:bead_packing} we depict slices through each 3D binary image and their signed Euclidean distance transforms, followed by plots of the bottleneck distances between persistence diagrams computed at different resolutions.  
In each case, we subsample the binary image to a lower resolution using the averaging technique described at the start of Section~\ref{sec:applications} with a threshold $t=0.5$. 
We only use kernel sizes that fit evenly within the image; with $N = 512$ we must use powers of two, $a = 2^k$.
As the persistence diagrams for these examples have so many points, we use an approximation algorithm for the bottleneck distances, as implemented in Giotto-TDA. Specifically, we use an approximation value of $\delta=0.1$ for the glass bead packing, and $\delta = 0.5$ for the other two to make these computations feasible.

The green curve shown on the plots of bottleneck distances in Fig.~\ref{fig:bead_packing} is the function $2\sqrt{3} r = 2\sqrt{3} (512/n)$, where $r =1$ is the voxel spacing for the  highest-resolution image, and $n$ is the resolution as measured by the number of voxels along each side of the cube.  
This is the bound on bottleneck distance we derived in Section~\ref{sec:PH_distance_transform} for the case that the voxel size $r < \mreach{X} / \sqrt{d}$.  
As already argued, the reach is zero for the glass bead pack, and likely to be zero for the sand pack, so the fact that this bound holds suggests that our estimation results are too generous, and/or that $\mleash{X}{s}$ is approximately $s$ despite the reach being zero for these examples.  

We note that the glass bead example has significantly larger bottleneck distances between its 1-dimensional persistence diagrams compared to the 0-dimensional distances for the resolutions $n=32, 64, 128, 256$. 
Inspection of the dimension one $\kPD{1}{D^n}$ diagrams and the original image shows that the larger values of $d_B(\kPD{1}{D^n}, \kPD{1}{D^{512}})$ are due to the presence of just a couple of high-persistence cycles near the boundary of the image, each cycle due to a bead that intersects two faces of the boundary. 
The sensitivity of the bottleneck distance to outliers is well known and is the reason Wasserstein distances between diagrams are often preferred. 

An important physical parameter associated with porous materials is the percolation threshold, $l_c$.  
This is the radius of the largest sphere that can pass through the pore space from one side of the image to the opposite. 
The distribution of points in $\kPD{0}{D^n}$ shows a clear signature of this critical length scale;  see~\cite{robins_percolating_2016} for details.  
As can be seen in the persistence diagrams 
in
Appendix~\ref{sec:app:matsci}, 
this signature yields the same estimate of $l_c$ for image resolutions $n=512,256, 128$ in the three example materials. 
This supports the designation of a $\Delta$-$\varepsilon$ plateau in the 0-dimensional bottleneck
distances with $\Delta = [64,512]$, and $\varepsilon = 10$ for all three samples.

\begin{figure}
    \centering
    Glass Bead Packing\\
    \includegraphics[width=0.58\columnwidth]{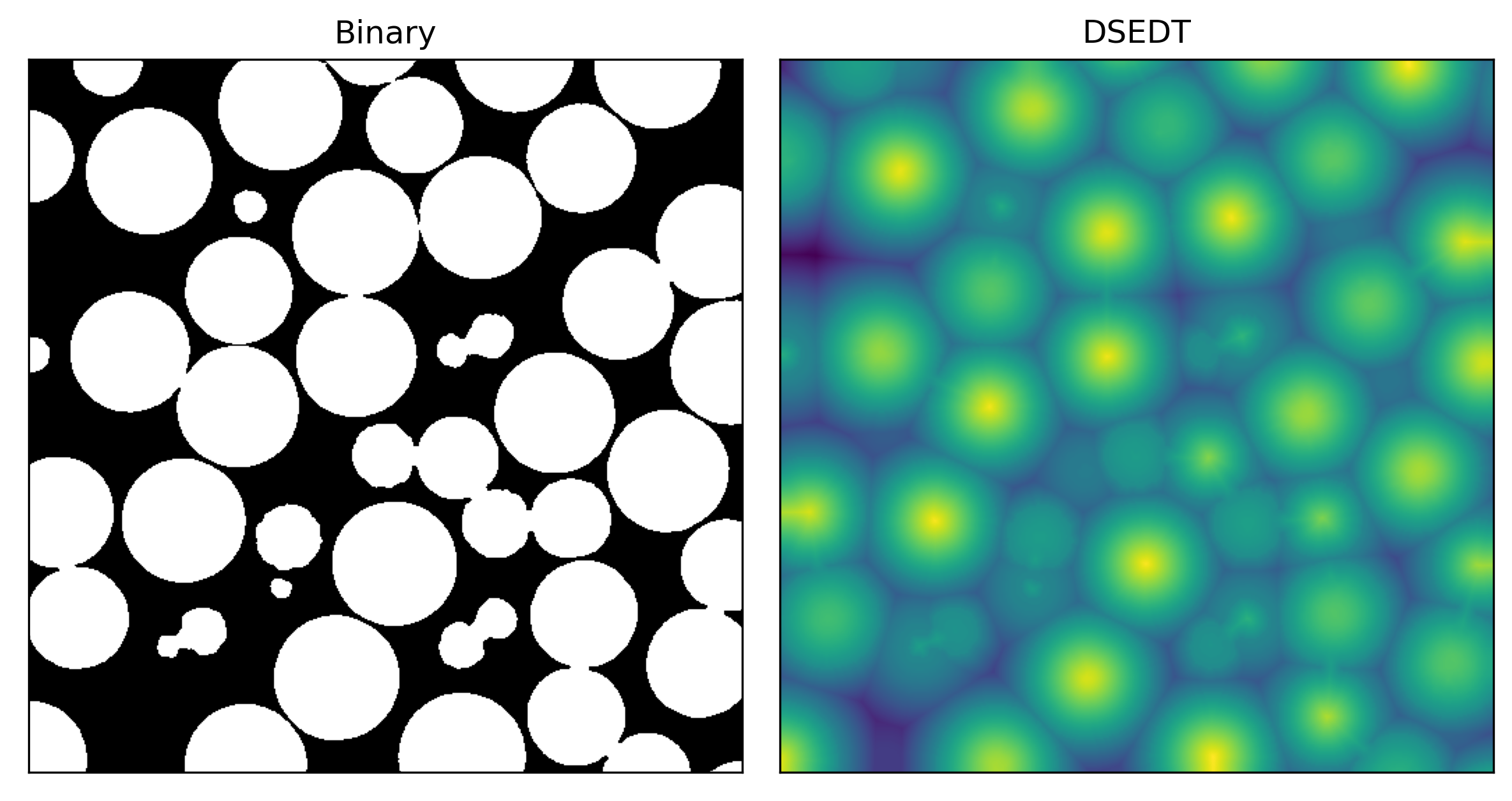}
    \includegraphics[width=0.40\columnwidth]{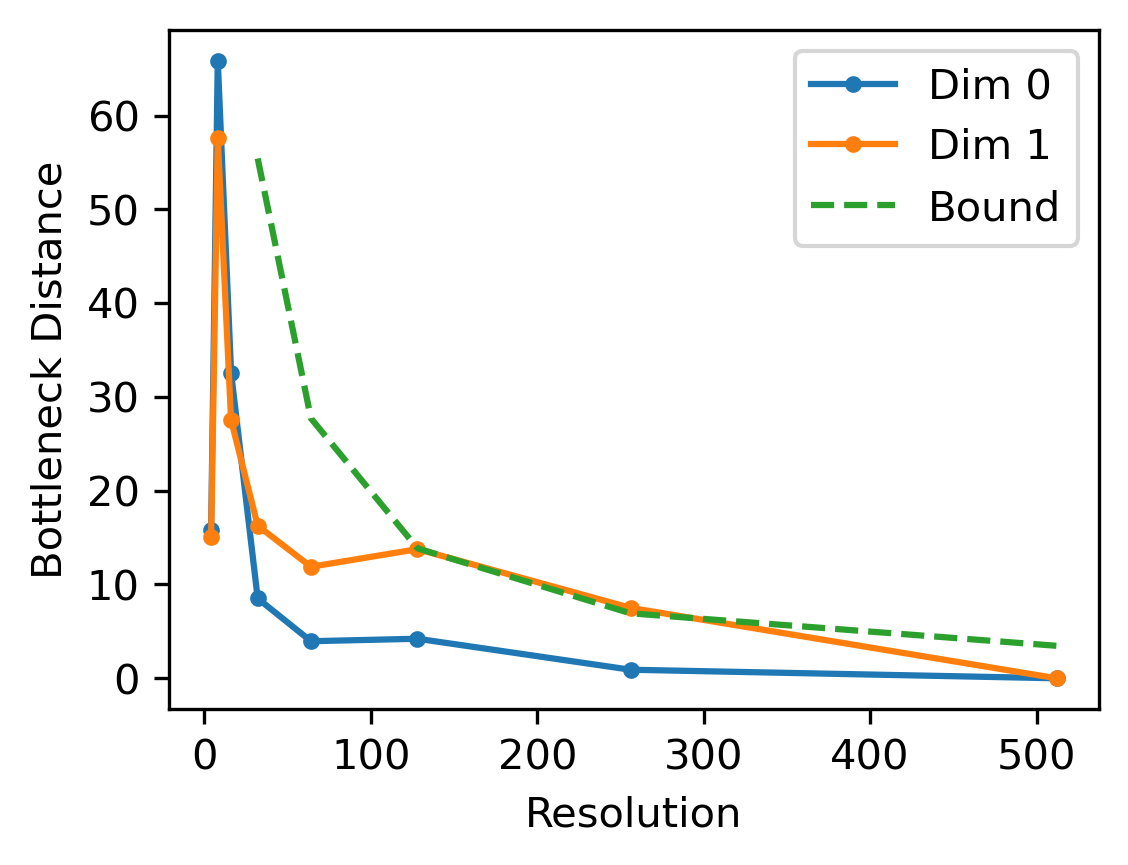} \\
    Castlegate Sandstone\\
    \includegraphics[width=0.58\columnwidth]{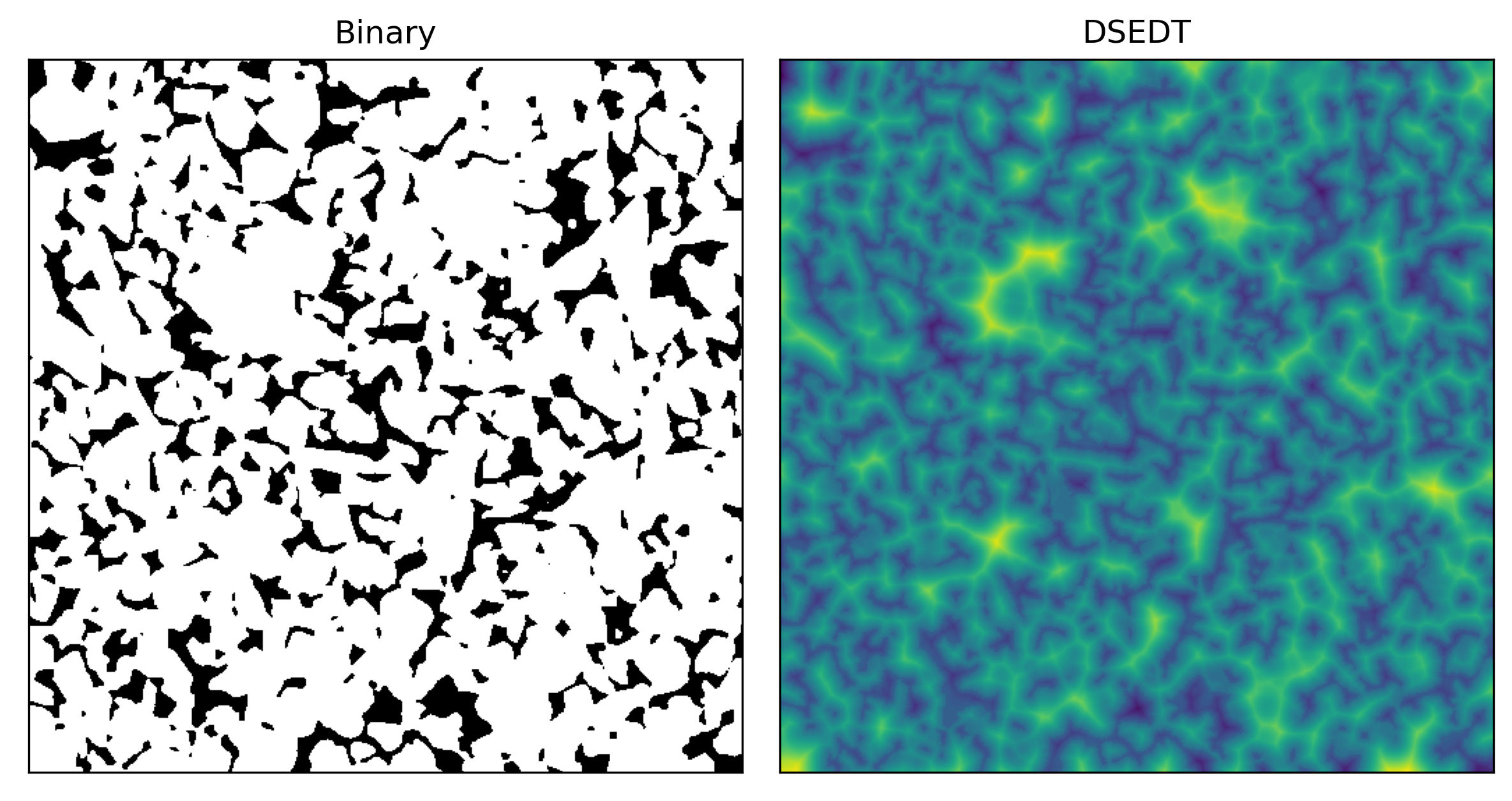}
    \includegraphics[width=0.40\columnwidth]{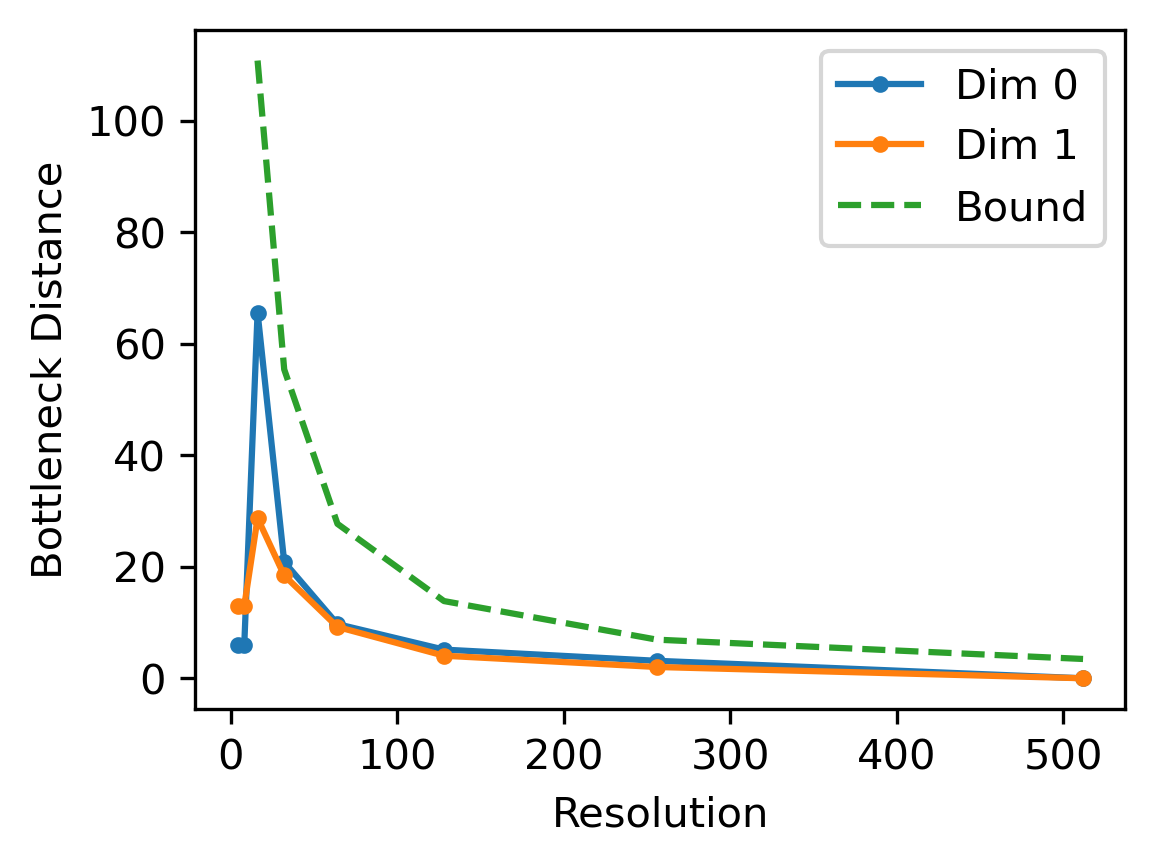} \\
    Sandpack\\
    \includegraphics[width=0.58\columnwidth]{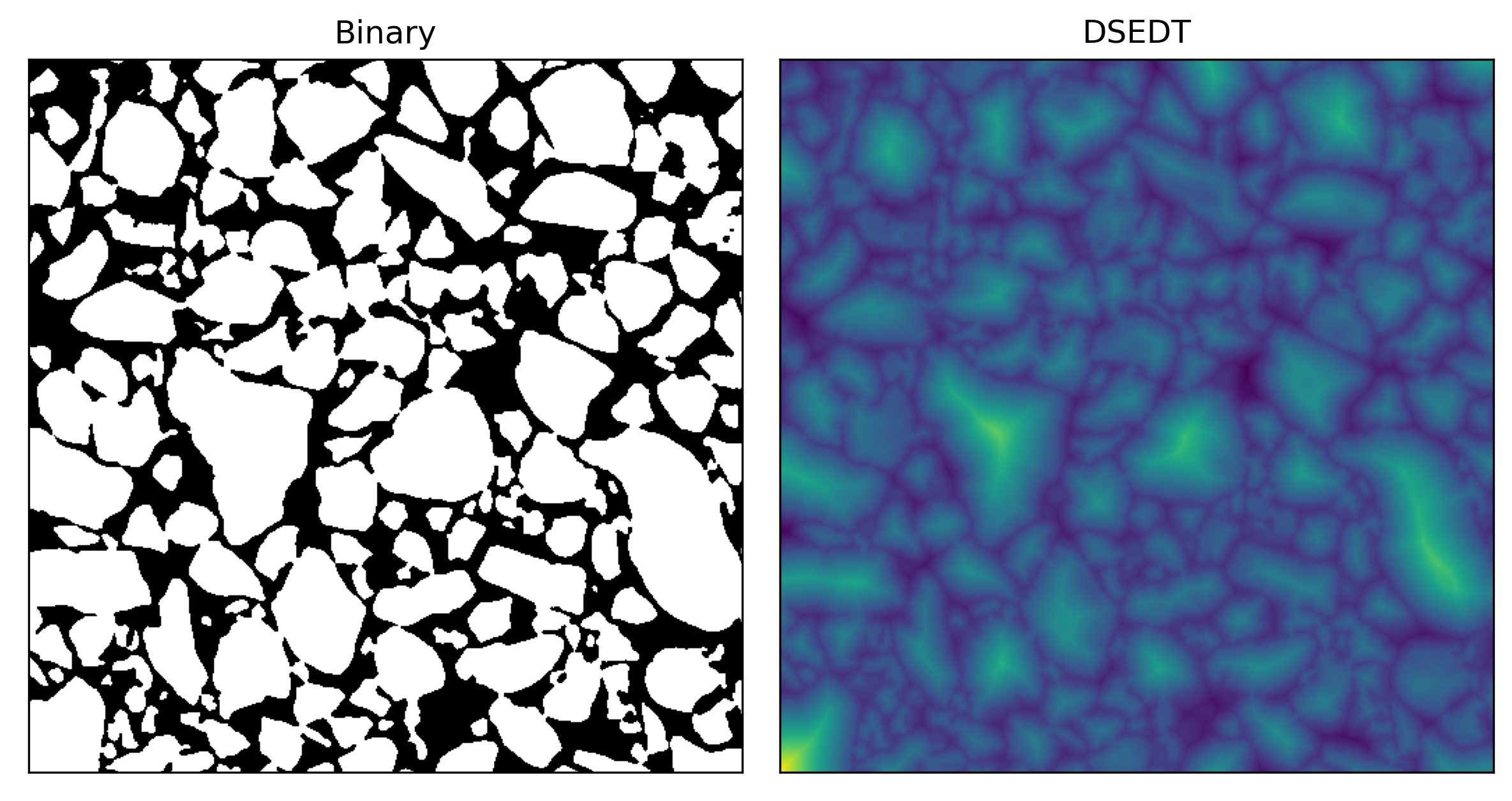}
    \includegraphics[width=0.40\columnwidth]{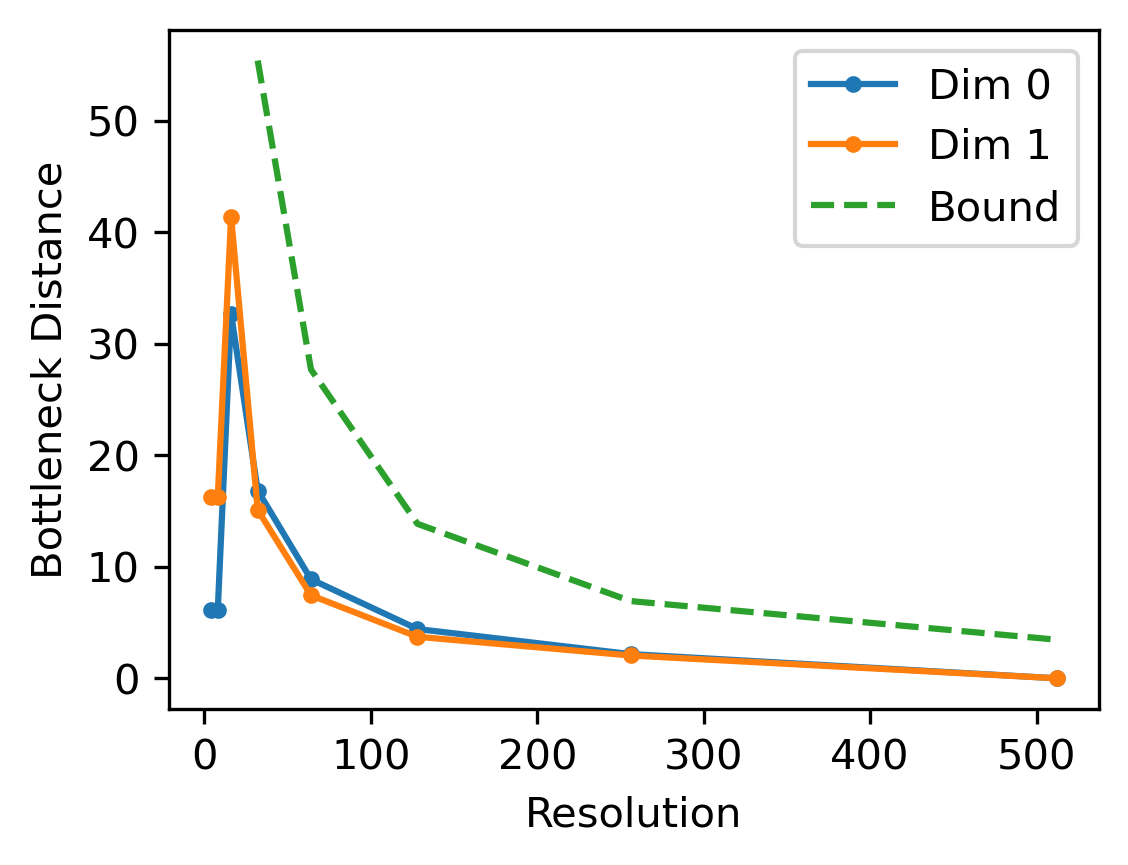}  
    \caption{2D slices of the binary and DSEDT images for the glass bead packing, Castlegate sandstone, and sand packing samples. The colormap for the DSEDT is scaled to the max and min values in each case. 
    The plots in the right column show the bottleneck distance between persistence diagrams for each downsampled resolution and the highest resolution.  The green dashed curve is the function $2\sqrt{3}(512/n)$, where $n$ is the image resolution. 
    }
    \label{fig:bead_packing}
\end{figure}

\section{Conclusion}\label{sec:conclusion}


This paper has presented two sets of results about the digital approximation of functions, and of solid objects. 
Section~\ref{sec:PH_grayscale} makes explicit how far points in a persistence diagram can move when working with locally averaged digital approximations to a function at different resolutions.  

Section~\ref{sec:PH_distance_transform} has considered a more subtle question of how close are the persistence diagrams of a solid object $X$ and a digital approximation to it, $X(r,t)$, when they are filtered by the continuous and discrete signed Euclidean distance transforms respectively.    
These results are the analogues of seminal work for point-cloud approximations of manifolds, which do not translate easily to the digital image setting. 

In general, the difference between the continuous and discrete distance transforms is given in terms of the voxel diameter plus the leash (Theorem~\ref{thm:main}).  
The leash may be large, for example, if $X$ is a material that has regions of micro-porosity or extended structures with geometric detail below the voxel size $r$.
When the voxel diameter $\sqrt{d}r < \mreach{X}$, we show that the leash, $\mleash{X}{\sqrt{d}r} = \sqrt{d}r$, so that the bottleneck distance between persistence diagrams is bounded by $2\sqrt{d} r$ (Corollary~\ref{cor:reach_main} and Remark~\ref{remark:tightening_cor}). 

The practical consequences of these results are that we expect the persistent homology to converge with increasing image resolution, but the error may not be monotonic, especially when considering images at low resolutions.  
When there is no prior information about the critical length scales of the object $X$, the x-ray density function given by the CT-scan can be interpreted as a (noisy) approximation to the density function $\rho_r$, and the Hausdorff distance between the two threshold choices $X(r,1-\epsilon)$ and $X(r,\epsilon)$ provides an estimate of the possible error in the persistence diagrams (Theorem~\ref{thm:Vanessas_bound}). 

Ultimately, our results provide guidance to practioners on how to balance the time, cost, and processing power required for image acquisition and persistent homology computations against the desired level of accuracy in their results.



\bibliographystyle{IEEEtran}
\bibliography{refs}

\begin{thebibliography}{10}
\providecommand{\url}[1]{#1}
\csname url@samestyle\endcsname
\providecommand{\newblock}{\relax}
\providecommand{\bibinfo}[2]{#2}
\providecommand{\BIBentrySTDinterwordspacing}{\spaceskip=0pt\relax}
\providecommand{\BIBentryALTinterwordstretchfactor}{4}
\providecommand{\BIBentryALTinterwordspacing}{\spaceskip=\fontdimen2\font plus
\BIBentryALTinterwordstretchfactor\fontdimen3\font minus
  \fontdimen4\font\relax}
\providecommand{\BIBforeignlanguage}[2]{{%
\expandafter\ifx\csname l@#1\endcsname\relax
\typeout{** WARNING: IEEEtran.bst: No hyphenation pattern has been}%
\typeout{** loaded for the language `#1'. Using the pattern for}%
\typeout{** the default language instead.}%
\else
\language=\csname l@#1\endcsname
\fi
#2}}
\providecommand{\BIBdecl}{\relax}
\BIBdecl

\bibitem{wildenschild_x-ray_2013}
D.~Wildenschild and A.~P. Sheppard, ``X-ray imaging and analysis techniques for
  quantifying pore-scale structure and processes in subsurface porous medium
  systems,'' \emph{Advances in Water Resources}, vol.~51, pp. 217--246, 2013.

\bibitem{herring_effect_2013}
A.~L. Herring, E.~J. Harper, L.~Andersson, A.~Sheppard, B.~K. Bay, and
  D.~Wildenschild, ``Effect of fluid topology on residual nonwetting phase
  trapping: {{Implications}} for geologic {{CO2}} sequestration,''
  \emph{Advances in Water Resources}, vol. 62, Part A, pp. 47--58, Dec. 2013.

\bibitem{herring_topological_2019}
A.~L. Herring, V.~Robins, and A.~P. Sheppard, ``Topological {{Persistence}} for
  {{Relating Microstructure}} and {{Capillary Fluid Trapping}} in
  {{Sandstones}},'' \emph{Water Resources Research}, vol.~55, no.~1, pp.
  555--573, Jan. 2019.

\bibitem{perssurvey}
H.~Edelsbrunner and J.~Harer, ``Persistent homology---a survey,'' in
  \emph{Surveys on discrete and computational geometry}, ser. Contemp.
  Math.\hskip 1em plus 0.5em minus 0.4em\relax Amer. Math. Soc., Providence,
  RI, 2008, vol. 453, pp. 257--282.

\bibitem{cub1}
V.~Robins, P.~J. Wood, and A.~P. Sheppard, ``Theory and algorithms for
  constructing discrete {Morse} complexes from grayscale digital images,''
  \emph{IEEE Transactions on Pattern Analysis and Machine Intelligence},
  vol.~33, pp. 1646--1658, 05 2011.

\bibitem{delgado-friedrichs_morse_2014}
O.~{Delgado-Friedrichs}, V.~Robins, and A.~Sheppard, ``Morse theory and
  persistent homology for topological analysis of {{3D}} images of complex
  materials,'' in \emph{2014 {{IEEE International Conference}} on {{Image
  Processing}} ({{ICIP}})}, Oct. 2014, pp. 4872--4876.

\bibitem{cohen2007stability}
D.~Cohen-Steiner, H.~Edelsbrunner, and J.~Harer, ``Stability of persistence
  diagrams,'' \emph{Discrete \& computational geometry}, vol.~37, no.~1, pp.
  103--120, 2007.

\bibitem{amenta_simple_2000}
\BIBentryALTinterwordspacing
N.~Amenta, S.~Choi, T.~K. Dey, and N.~Leekha, ``A {Simple} {Algorithm} for
  {Homeomorphic} {Surface} {Reconstruction},'' in \emph{Proceedings of the
  {Sixteenth} {Annual} {Symposium} on {Computational} {Geometry}}, ser. {SCG}
  '00.\hskip 1em plus 0.5em minus 0.4em\relax ACM, 2000, pp. 213--222.
  [Online]. Available: \url{http://doi.acm.org/10.1145/336154.336207}
\BIBentrySTDinterwordspacing

\bibitem{niyogi_finding_2008}
\BIBentryALTinterwordspacing
P.~Niyogi, S.~Smale, and S.~Weinberger, ``\BIBforeignlanguage{en}{Finding the
  {Homology} of {Submanifolds} with {High} {Confidence} from {Random}
  {Samples}},'' \emph{\BIBforeignlanguage{en}{Discrete Comput Geom}}, vol.~39,
  no. 1-3, pp. 419--441, Mar. 2008. [Online]. Available:
  \url{http://link.springer.com/10.1007/s00454-008-9053-2}
\BIBentrySTDinterwordspacing

\bibitem{chazal_sampling_2009}
\BIBentryALTinterwordspacing
F.~Chazal, D.~Cohen-Steiner, and A.~Lieutier, ``\BIBforeignlanguage{en}{A
  {Sampling} {Theory} for {Compact} {Sets} in {Euclidean} {Space}},''
  \emph{\BIBforeignlanguage{en}{Discrete Comput Geom}}, vol.~41, no.~3, pp.
  461--479, Apr. 2009. [Online]. Available:
  \url{http://link.springer.com/10.1007/s00454-009-9144-8}
\BIBentrySTDinterwordspacing

\bibitem{VR_approx}
D.~Attali, A.~Lieutier, and D.~Salinas, ``Vietoris-rips complexes also provide
  topologically correct reconstructions of sampled shapes,'' \emph{27th Annual
  Symposium on Computational Geometry}, vol.~46, 05 2013.

\bibitem{kim2020homotopy}
J.~Kim, J.~Shin, F.~Chazal, A.~Rinaldo, and L.~Wasserman, ``Homotopy
  reconstruction via the cech complex and the vietoris-rips complex,'' in
  \emph{36th International Symposium on Computational Geometry (SoCG
  2020)}.\hskip 1em plus 0.5em minus 0.4em\relax Schloss
  Dagstuhl-Leibniz-Zentrum f{\"u}r Informatik, 2020.

\bibitem{bendich_computing_2010}
P.~Bendich, H.~Edelsbrunner, and M.~Kerber, ``Computing robustness and
  persistence for images,'' \emph{IEEE Transactions on Visualization and
  Computer Graphics}, vol.~16, pp. 1251--1260, 2010.

\bibitem{dlotko_rigorous_2018}
P.~D{\l}otko and T.~Wanner, ``Rigorous cubical approximation and persistent
  homology of continuous functions,'' \emph{Computers \& Mathematics with
  Applications}, vol.~75, no.~5, pp. 1648--1666, Mar. 2018.

\bibitem{botha_mapping_2016}
P.~W. Botha and A.~P. Sheppard, ``Mapping permeability in low-resolution
  micro-{{CT}} images: A multiscale statistical approach,'' \emph{Water
  Resources Research}, vol.~52, no.~6, pp. 4377--4398, Jun. 2016.

\bibitem{huang_effect_2021}
R.~Huang, A.~L. Herring, and A.~Sheppard, ``Effect of {{Saturation}} and
  {{Image Resolution}} on {{Representative Elementary Volume}} and
  {{Topological Quantification}}: An {{Experimental Study}} on {{Bentheimer
  Sandstone Using Micro}}-{{CT}},'' \emph{Transport in Porous Media}, vol. 137,
  no.~3, pp. 489--518, Apr. 2021.

\bibitem{duality}
B.~Bleile, A.~Garin, T.~Heiss, K.~Maggs, and V.~Robins, ``The persistent
  homology of dual digital image constructions,''
  \emph{https://arxiv.org/abs/2102.11397v1}, 2021.

\bibitem{blum}
H.~Blum, ``A transformation for extracting descriptors of shape,'' in
  \emph{Models for the {{Perception}} of {{Speech}} and {{Visual
  Forms}}}.\hskip 1em plus 0.5em minus 0.4em\relax {MIT Press}, 1967, pp.
  362--380.

\bibitem{federer1959curvature}
H.~Federer, ``Curvature measures,'' \emph{Transactions of the American
  Mathematical Society}, vol.~93, no.~3, pp. 418--491, 1959.

\bibitem{thale_50_2008}
C.~Th\"ale, ``\BIBforeignlanguage{en}{50 {{Years}} of sets with positive reach,
  a survey.}'' \emph{\BIBforeignlanguage{en}{Surveys in Mathematics and its
  Applications}}, vol.~3, pp. 123--165, 2008.

\bibitem{tauzin2020giottotda}
G.~Tauzin, U.~Lupo, L.~Tunstall, J.~B. Pérez, M.~Caorsi, A.~Medina-Mardones,
  A.~Dassatti, and K.~Hess, ``giotto-tda: A topological data analysis toolkit
  for machine learning and data exploration,'' 2020.

\bibitem{phimages_githubrepo}
``{PH}-of-{Images} {G}it{H}ub {R}epository,''
  \url{https://github.com/sarahtymochko/PH-of-Images/}.

\bibitem{sheppard_network_2005}
A.~Sheppard and M.~Prodanovic, ``Network generation comparison forum,''
  \url{http://www.digitalrocksportal.org/projects/16}, 2015.

\bibitem{robins_percolating_2016}
V.~Robins, M.~Saadatfar, O.~{Delgado-Friedrichs}, and A.~P. Sheppard,
  ``Percolating length scales from topological persistence analysis of
  micro-{{CT}} images of porous materials,'' \emph{Water Resources Research},
  vol.~52, no.~1, pp. 315--329, Jan. 2016.

\end{thebibliography}

\appendices

\section{Federer's Lemma about the reach} \label{sec:app:federer}
This section describes a lemma about the reach, proven by Federer in \cite[Theorem 4.8 (12)]{federer1959curvature}, stating that when walking from a point $a \in \partial A$ orthogonally away from a closed set $A$ for distance $r<\reach(A)$, then the closest point of $A$ is still the starting point $a$, and thus the distance to $A$ is $r$. Before we can state this rigorously in Lemma~\ref{lemma:federer_rigorous}, we need formal definitions \cite[Definitions 4.3 and 4.4]{federer1959curvature} of tangent vectors and normal vectors:

\begin{definition}[Tangent vector] \label{def:tangent}
Let $A \subseteq \R^d$ be closed and $a\in \partial A$. Then $u \in \R^d$ is a \emph{tangent vector of $A$ at $a$} if either $u=0$ or for every $\varepsilon > 0$ there exists a point $b \in A$ with
\begin{align*}
    0 < \norm{b-a} < \varepsilon \textrm{ and } 
    \norm{ \frac{b-a}{\norm{b-a}} - \frac{u}{\norm{u}} } < \varepsilon.
\end{align*}
\end{definition}

\begin{definition}[Normal vector] \label{def:normal}
Let $A \subseteq \R^d$ be closed and $a\in \partial A$. Then $v \in \R^d$ is a \emph{normal vector of $A$ at $a$} if for every tangent vector $u$ of $A$ at $a$, the scalar product $v \cdot u$ is non-positive. 
\end{definition}

\begin{lemma}[Federer] \label{lemma:federer_rigorous}
Let $A \subseteq \R^d$ be closed and $a\in \partial A$. Let $\reach(A)>r>0$.
Let $v$ be a normal vector of $A$ at $a$. Then,
\begin{align*}
    d(a + r \frac{v}{\norm{v}}, A) = d(a + r \frac{v}{\norm{v}}, a)=r.
\end{align*}
\end{lemma}


\section{Tighter Bounds} 
\label{sec:app:tightening}
As mentioned in 
Remark~\ref{remark:tightening_cor},
the bound on the bottleneck distance in 
Corollary~\ref{cor:reach_main} 
can be tightened from $3 \sqrt{d} r$ to $2 \sqrt{d} r$. 
To prove this, we first need to prove the tighter version of 
Lemma~\ref{lemma:generalizeCS-E-H}
mentioned in 
Remark~\ref{remark:tighteningCS-E-H}:
\begin{lemma} 
\label{lemma:app:remark:tighteningCS-E-H}
Let $A \subseteq \R^d$ have boundary with positive reach, and let $B \subseteq \R^d$. \\
If $\max \{ \dH(A,B),\dH(A^\compl,B^\compl) \} < \mreach{A}$, then
\begin{align*}
    \norminf{\dt{A}-\dt{B}} \leq \max \{ \dH(A,B), \ \dH(A^\compl,B^\compl) \}.
\end{align*} 
\end{lemma}
\begin{proof}
Similarly to the proof of 
Lemma~\ref{lemma:generalizeCS-E-H}, 
we distinguish between 4 different cases for an arbitrary point $p$.
Cases 1 and 2 stay unchanged.

\emph{Case 3: $p \in A ,\ p \notin B$.}
As in the proof of 
Lemma~\ref{lemma:generalizeCS-E-H}, 
we need to bound $\absolute{\dt{A}(p)-\dt{B}(p)} = \absolute{-d(A^\compl,p)-d(p,B)} = d(A^\compl,p) + d(p,B)$.
To prove Case 3, we distinguish between two sub-cases.

\emph{Case 3a: $p \in \partial A \cap A ,\ p \notin B$.}
As $p\in \partial A$, there is a sequence in $A^\compl$ converging to $p$, proving $d(A^\compl,p)=0$.
Thus,
\begin{align*}
\absolute{\dt{A}(p)-\dt{B}(p)} &= d(A^\compl,p) + d(p,B) \leq 0 + \dH(A,B) \\
&\leq \max \{ \dH(A,B), \ \dH(A^\compl,B^\compl) \}.
\end{align*}

\emph{Case 3b: $p \in \inter(A) ,\ p \notin B$.}
Let $a \in \partial \cl(A^\compl)$ be such that $d(A^\compl,p) = d(\cl(A^\compl),p) = d(a,p)$. 
As $a \neq p$, the vector $p-a$ has non-zero length. From $d(\cl(A^\compl),p) = d(a,p)$ follows that $n_a = \frac{p-a}{\norm{p-a}}$ is a unit normal vector of $\cl(A^\compl)$ at $x$ (for a rigorous definition of normal vector, see 
Definition~\ref{def:normal}).
Let $\varepsilon \in (0,\mreach{A}-\dH(A^\compl,B^\compl))$.
Let $x = a + (\dH(A^\compl,B^\compl) + \varepsilon) n_a$. 
As $\dH(A^\compl,B^\compl) + \varepsilon < \reach(\partial A) \leq \reach(\cl(A^\compl))$, 
Lemma~\ref{lemma:federer_rigorous} 
yields 
\begin{align*}
d(A^\compl, x) &= d(\cl(A^\compl), x) = d(a, x) = \dH(A^\compl,B^\compl) + \varepsilon  \\
&> \dH(A^\compl,B^\compl).
\end{align*}
Hence, by the definition of Hausdorff distance, $x$ cannot be in $B^\compl$ and thus $x \in B$. 
The points $a$, $p$, and $x$ lie on a straight line by construction. To prove $d(a,p) < d(a,x)$, let us assume $d(a,p) \geq d(a,x)$ 
which gives a contradiction between $d(A^\compl,p) = d(a,p) \geq d(a,x) > \dH(A^\compl,B^\compl)$ and $p \in B^\compl$. Therefore, $p$ lies on the line segment between $a$ and $x$.
With this we can bound $d(A^\compl,p)+d(p,B) \leq d(a,p)+d(p,x) = d(a,x) = \dH(A^\compl,B^\compl) + \varepsilon$. 
As this bound is true for every $\varepsilon \in (0,\reach(A^\compl)-\dH(A^\compl,B^\compl))$, we follow
\begin{align*}
\absolute{\dt{A}(p)-\dt{B}(p)} &= d(A^\compl,p)+d(p,B) \leq \dH(A^\compl,B^\compl) \\
&\leq \max \{ \dH(A,B), \ \dH(A^\compl,B^\compl) \}.
\end{align*}

\emph{Case 4: $p \notin A ,\ p \in B$.}
Analogously $\absolute{\dt{A}(p)-\dt{B}(p)} \leq \max \{ \dH(A,B), \ \dH(A^\compl,B^\compl) \}$.
\end{proof}

With this we can prove the tighter version of 
Corollary~\ref{cor:reach_main}, 
mentioned in 
Remark~\ref{remark:tightening_cor}:
\begin{corollary} 
\label{cor:app:tighten}
Using the notation of 
Theorem~\ref{thm:main}.
If $r < \frac{1}{\sqrt{d}} \mreach{X}$, then 
\begin{align*}
\dB(\PD{\dt{X}}, \PD{\tg}) \leq 2 \sqrt{d} r
\end{align*}
\end{corollary}
\begin{proof}
Note that $\max \{ \dH(X,\tcx), \ \dH(X^\compl,\tcx^\compl) \}$ is the maximum of the $4$ suprema from 
Lemma~\ref{lemma:bounding_suprema_by_leash}.
We thus use 
Lemma~\ref{lemma:bounding_suprema_by_leash} 
to bound $\max \{ \dH(X,\tcx), \ \dH(X^\compl,\tcx^\compl) \}$ by $\max \{ \mleash{X}{\sqrt{d}r} , \sqrt{d}r \}$, which is $\sqrt{d}r$ by 
Lemma~\ref{lemma:reach_leash}.


Similar to the proof of 
Lemma~\ref{lemma:bound_with_suprema}, 
we combine the stability theorem of persistent homology \cite{cohen2007stability}, the triangle inequality, and 
Lemma~\ref{lemma:DSEDTvsCSEDT}, 
now with the new Lemma~\ref{lemma:app:remark:tighteningCS-E-H}:
\begin{align*}
&\dB(\PD{\dt{X}}, \PD{\tg})\\
&\leq \norminf{\dt{X} - \tg} \\
&\leq \norminf{\dt{X} - \dt{\tcx}} + \norminf{ \dt{\tcx} - \tg} \\
&\leq \max \{ \dH(X,\tcx), \ \dH(X^\compl,\tcx^\compl) \} + \sqrt{d} r \\
&\leq 2 \sqrt{d} r.
\end{align*} 
\end{proof}

\section{Bottleneck Distances with Respect to Pixel Size}
\label{sec:app:dists_pixelsize}

For all our examples we present the plots of bottleneck distances vs.\ image resolution which can exhibit the $\Delta$-$\varepsilon$ plateau behavior described in Section~\ref{sec:theory_not_apply}. 
However there is another way of visualizing this, considering bottleneck distance vs. pixel size, defined in Section~\ref{subsec:images_cont_func} and denoted by $r$ throughout the paper.
Figures~\ref{fig:small_dot_ps}, \ref{fig:nested_rings_ps}, and \ref{fig:bead_packing_ps} present this alternate perspective on Figures~\ref{fig:small_dot}, \ref{fig:nested_rings}, and \ref{fig:bead_packing}, respectively.

The important thing to note is that in the figures presented in terms of resolution, the $x$-axis goes from low to high resolution from left to right. 
However, in the figures presented here in terms of pixel size, the highest resolution image (considered our ``ground truth'') has a pixel size of $1$, thus the $x$-axis in these plots go from high to low resolution from left to right. 
Thus, the plateau behavior we look for occurs from right to left in the plots with respect to pixel size.

\begin{figure}
    \centering
    \includegraphics[align=c,width=0.99\columnwidth]{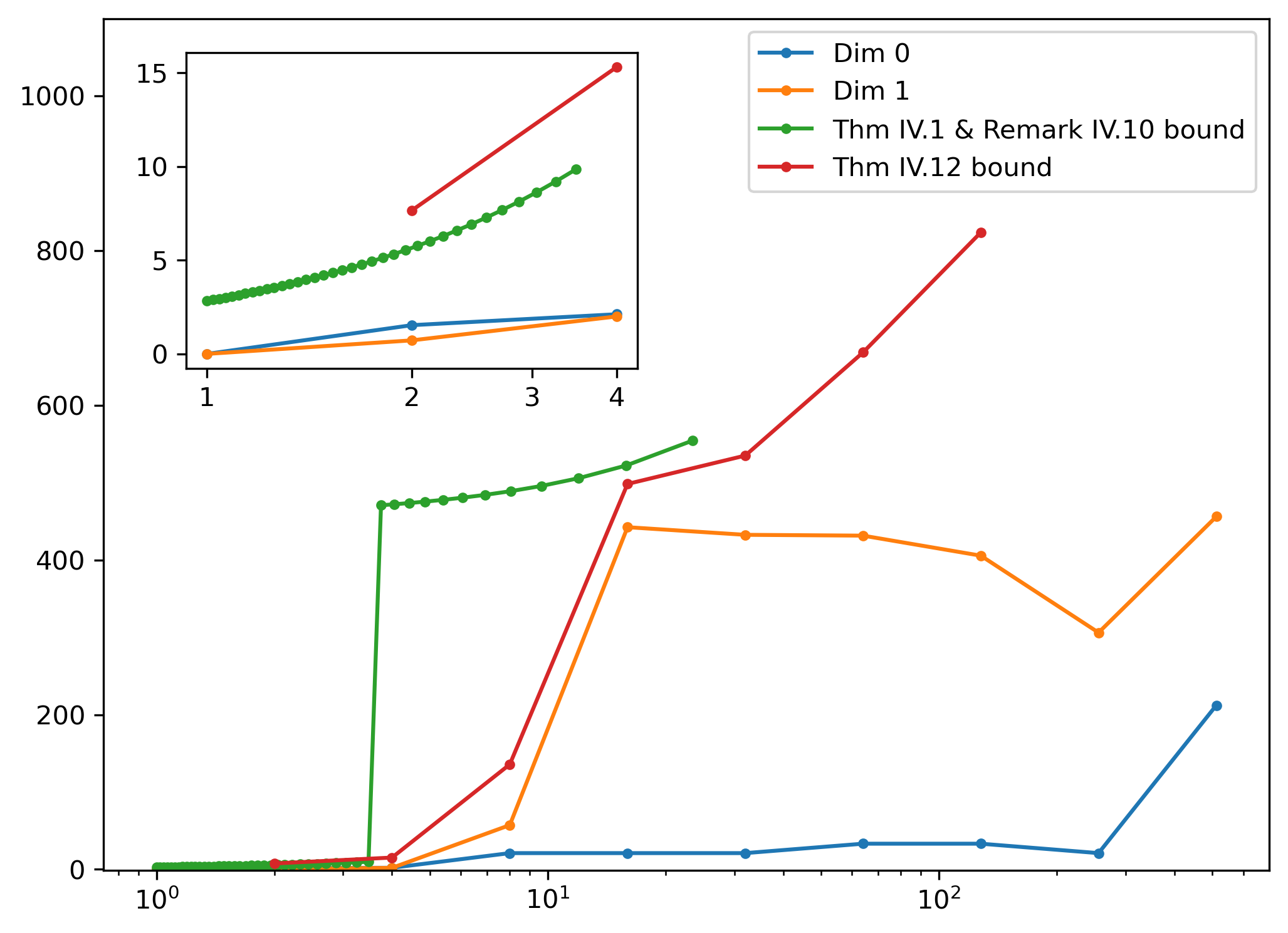}    \\
    \caption{ 
    Bottleneck distances comparing persistence diagrams for each pixel size with the image with pixel size of 1. The distances for dimensions 0 and 1 are shown along with the bounds provided by Theorem~\ref{thm:main} with Remark~\ref{remark:tightening_cor} (green) and Theorem~\ref{thm:Vanessas_bound} (red).
    Note that the $x$-axis is on a logarithmic scale. 
    This figure contains the same information as Fig.~\ref{fig:small_dot} but presented in terms of pixel size instead of resolution.}  
    \label{fig:small_dot_ps}
\end{figure}

\begin{figure}
    \centering
    \includegraphics[width=0.99\columnwidth]{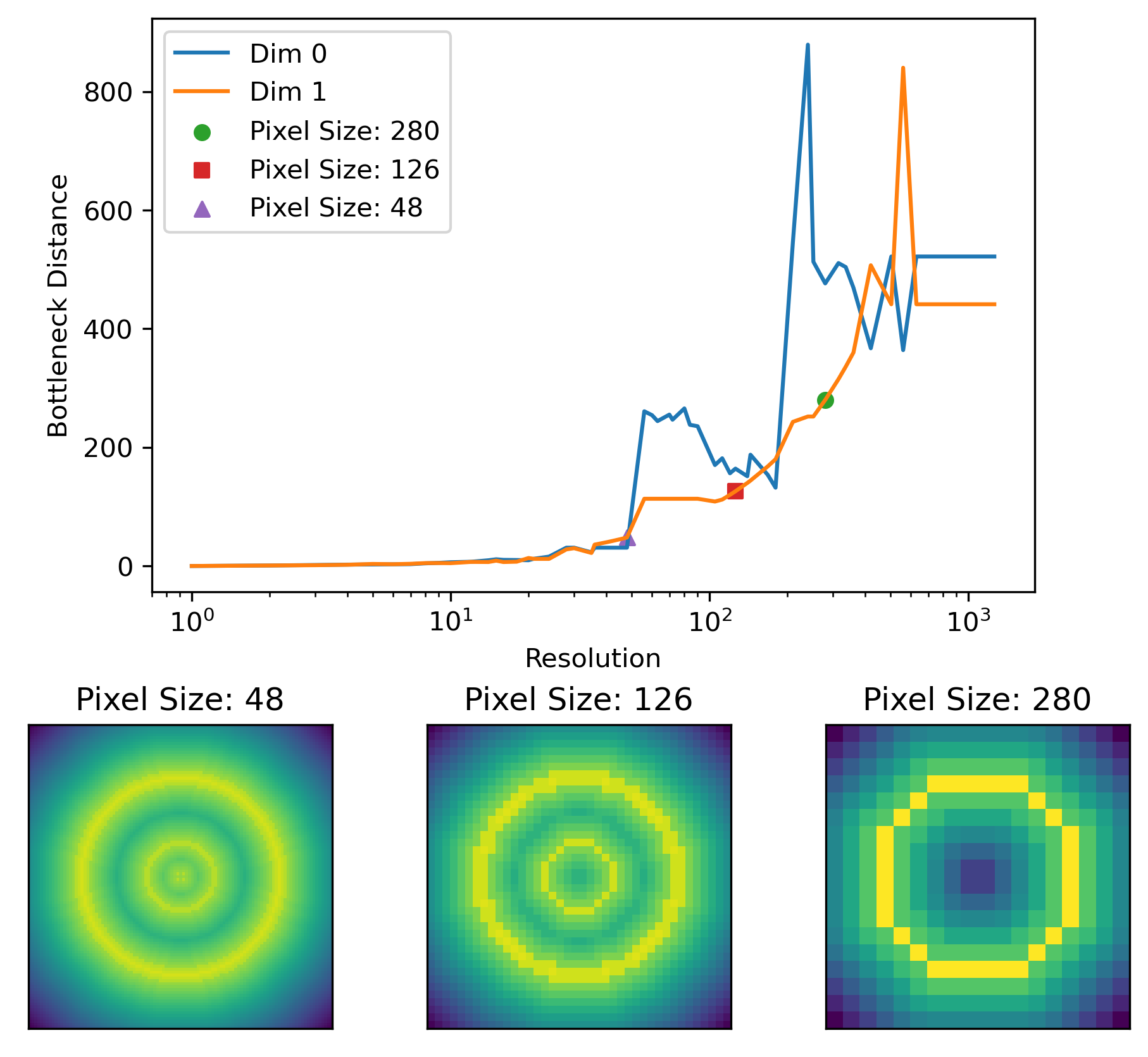}
    \caption{Top: the plot of the bottleneck distances between 1-dimensional persistence diagrams from each pixel size and the image with pixel size of 1. 
    Note that the $x$-axis is on a logarithmic scale. 
    Bottom: the DSEDT of the downsampled image at three different pixel sizes. This figure contains the same information as Fig.~\ref{fig:nested_rings} but presented in terms of pixel size instead of resolution.}
    \label{fig:nested_rings_ps}
\end{figure}

\begin{figure}
    \centering
    Glass Bead Packing \\ 
    \includegraphics[width=0.8\columnwidth]{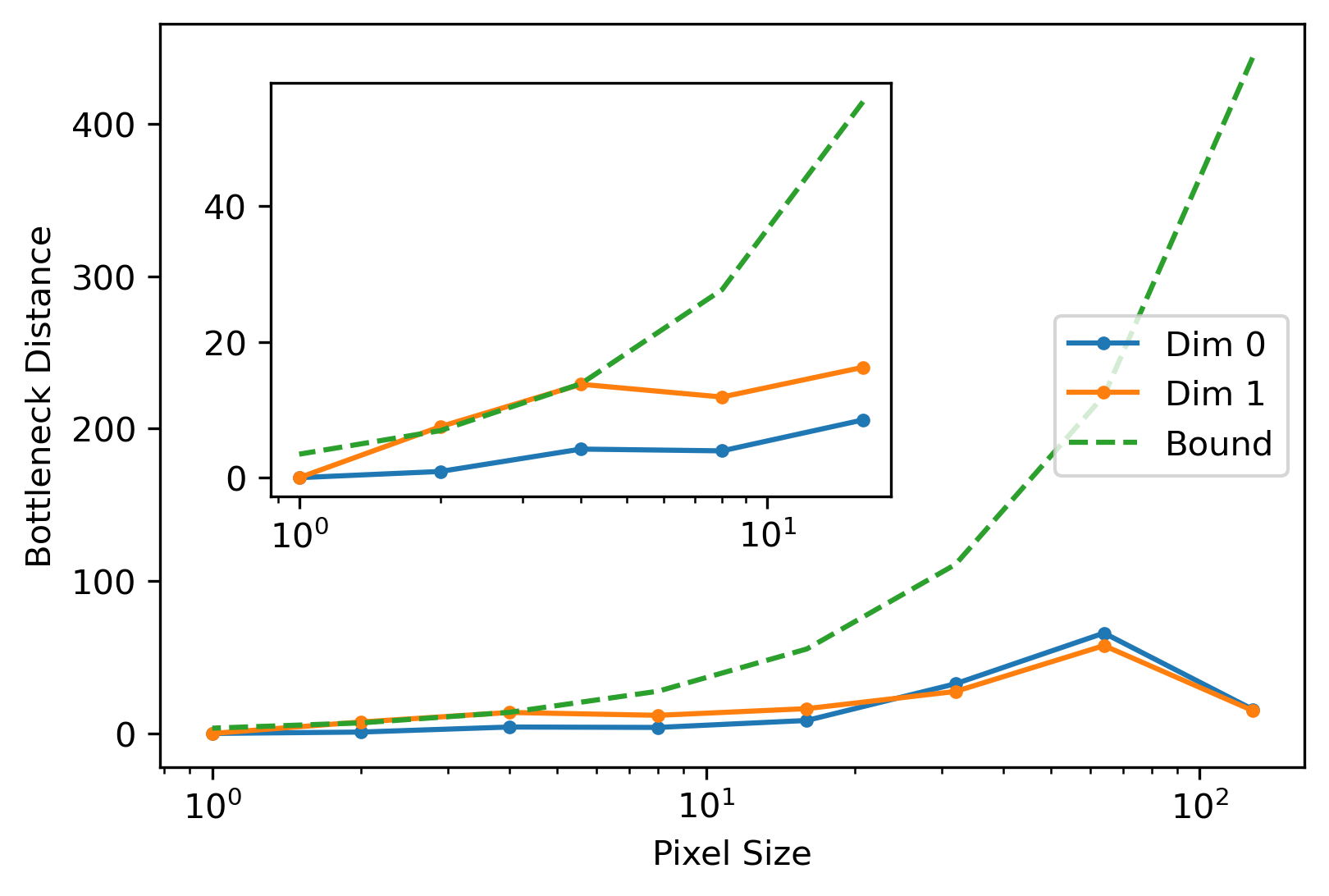}  \\
    Castlegate Sandstone\\
    \includegraphics[width=0.8\columnwidth]{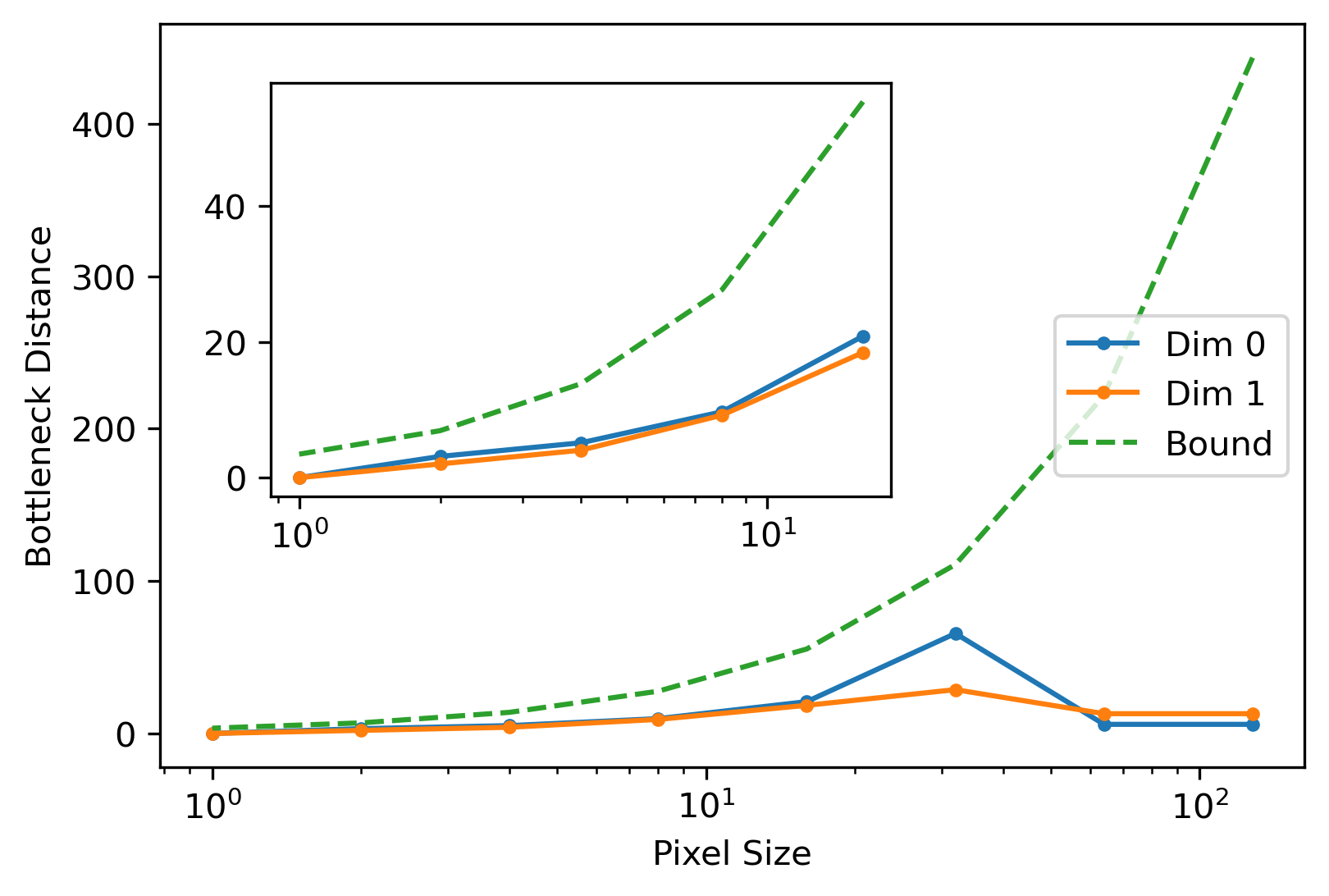}  \\
    Sandpack\\
    \includegraphics[width=0.8\columnwidth]{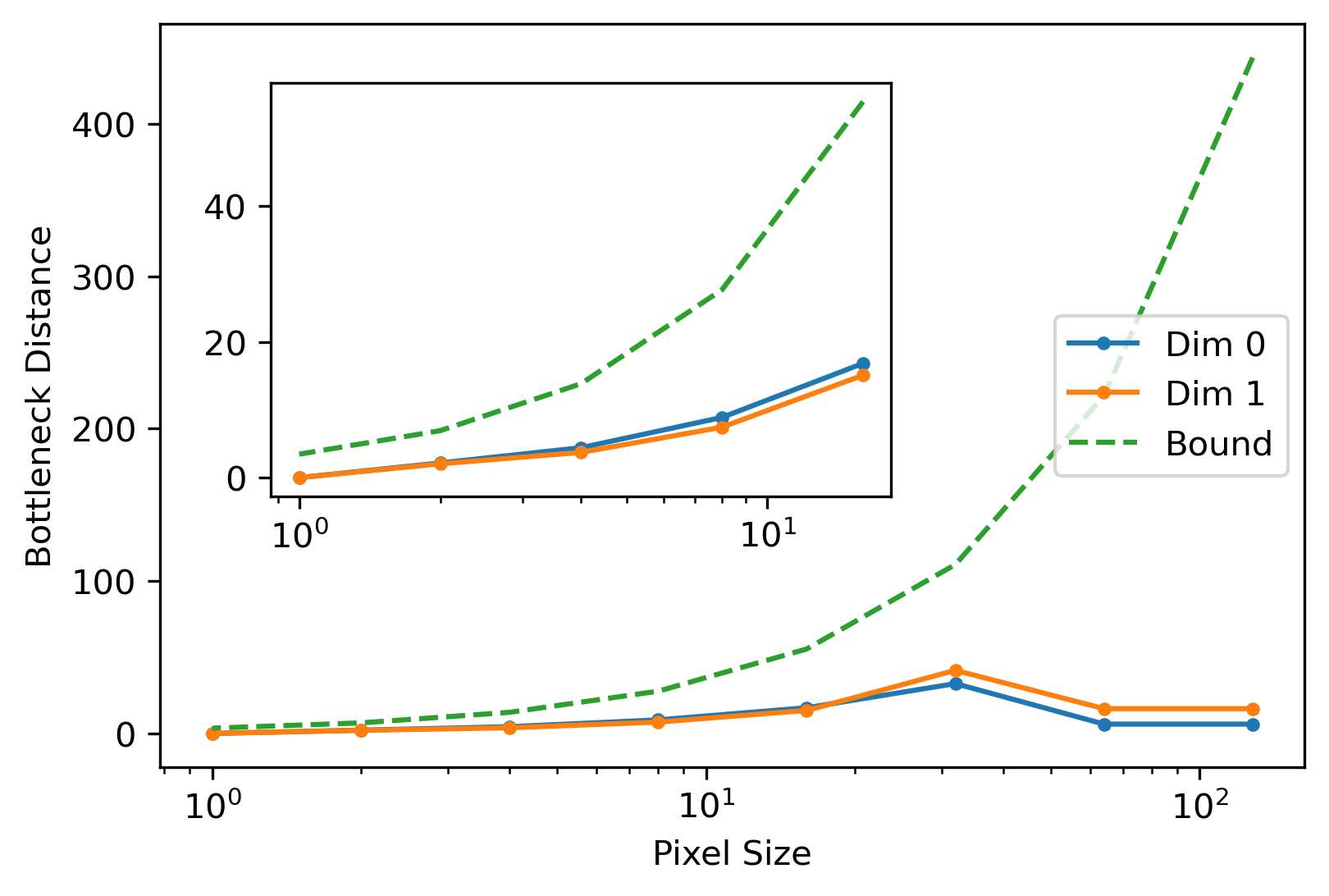}  
    \caption{Bottleneck distance between persistence diagrams for each larger pixel size and the image with pixel size of 1 for each of the three material science examples. The green dashed curve is the function $2r\sqrt{3}$, where $r$ is the pixel size. 
    Note that the $x$-axes are on a logarithmic scale. 
    This figure contains the same information as Fig.~\ref{fig:bead_packing} but presented in terms of pixel size instead of resolution.
    }
    \label{fig:bead_packing_ps}
\end{figure}


\section{Persistence Diagrams from Material Science Examples}
\label{sec:app:matsci}

Section~\ref{sec:matsci} 
presents three material science examples: glass bead packing, a Castlegate sandstone sample, and a sand packing sample.  
The original binary images are available from~\cite{sheppard_network_2005}. 
In Figures~\ref{fig:beadpack},~\ref{fig:castlegate}, and~\ref{fig:sandpack}, we show the 0- and 1-dimensional 
persistence diagrams for the 3D images computed using several resolutions for the bead packing, Castlegate sandstone, and sand packing, respectively.

As shown in ~\cite{robins_percolating_2016}, the percolation threshold, $l_c$, can be determined from the distribution of points in the zero-dimensional persistence diagram. 
This threshold is the radius of the largest sphere that can pass through the pore space from one side of the image to the opposite and is an important physical parameter associated with porous materials.

In the three sets of examples shown in Figures~\ref{fig:beadpack}-\ref{fig:sandpack}, the distribution of points in $\kPD{0}{D_n}$ shows a clear signature of this critical length scale.
This signature yields the same estimate for $l_c$ for image resolutions $n=512$, $256$, and $128$ in the three example materials. 

\begin{figure*}[p]
    \centering
    \large{Bead Packing} \\
    \includegraphics[width=0.8\textwidth]{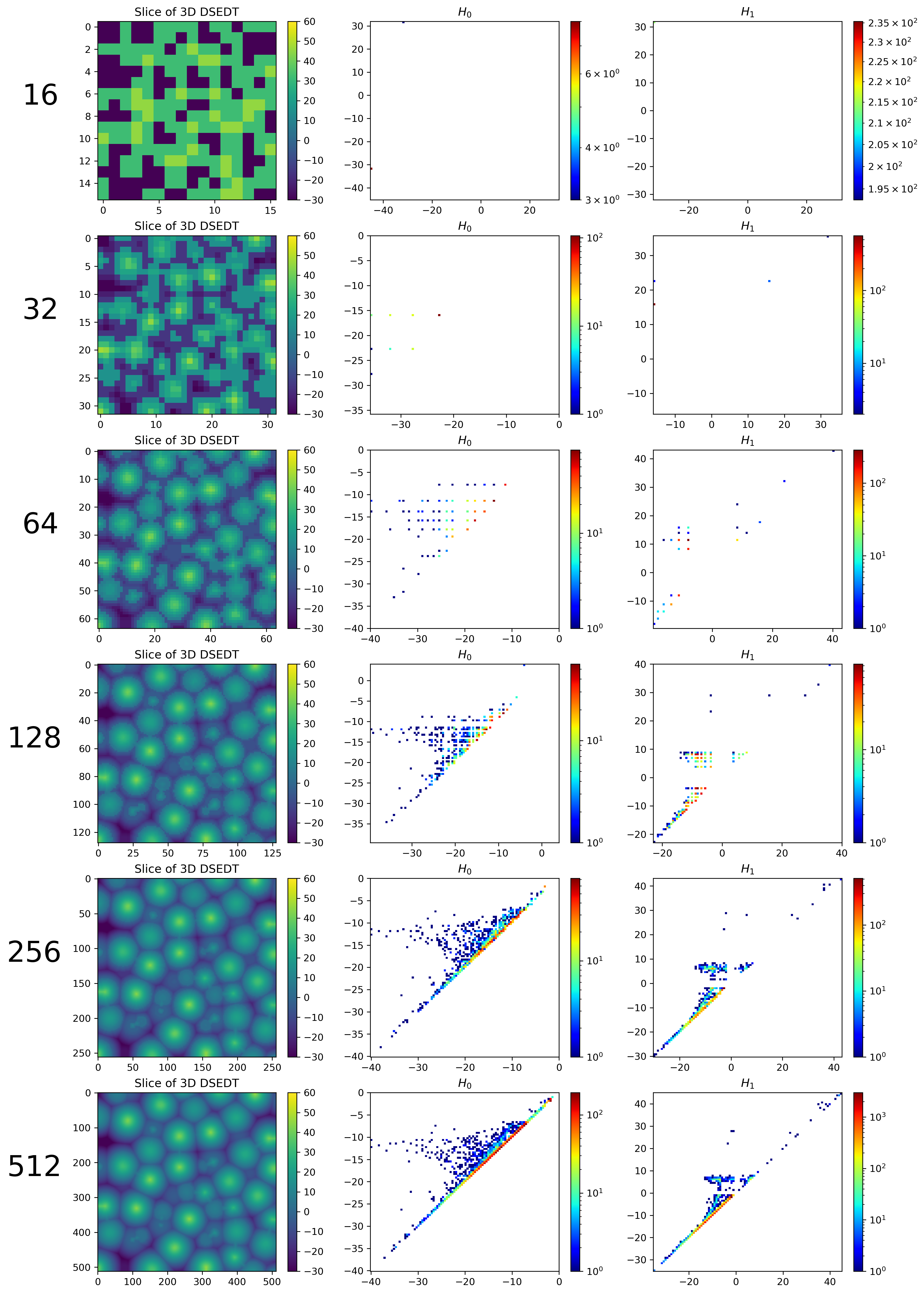}
    \caption{Each row shows a slice of the 3D SEDT, and the 0 and 1-dimensional persistence diagrams of the bead packing sample at different resolutions.}
    \label{fig:beadpack}
\end{figure*}

\begin{figure*}[p]
    \centering
    \large{Castlegate Sandstone} \\
    \includegraphics[width=0.8\textwidth]{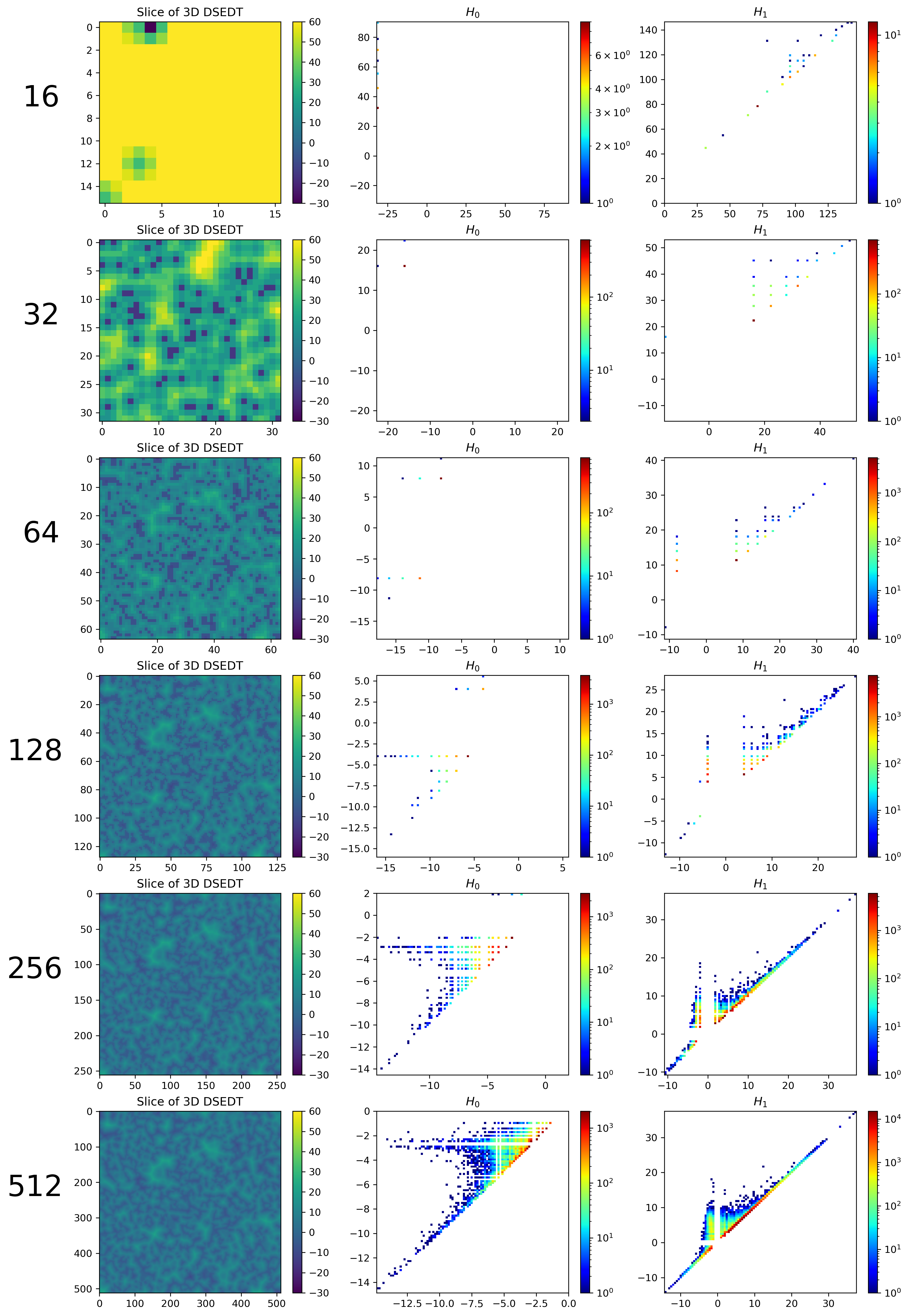}
    \caption{Each row shows a slice of the 3D SEDT, and the 0 and 1-dimensional persistence diagrams of the Castlegate sandstone sample at different resolutions.}
    \label{fig:castlegate}
\end{figure*}

\begin{figure*}[p]
    \centering
    \large{Sand Pack} \\
    \includegraphics[width=0.8\textwidth]{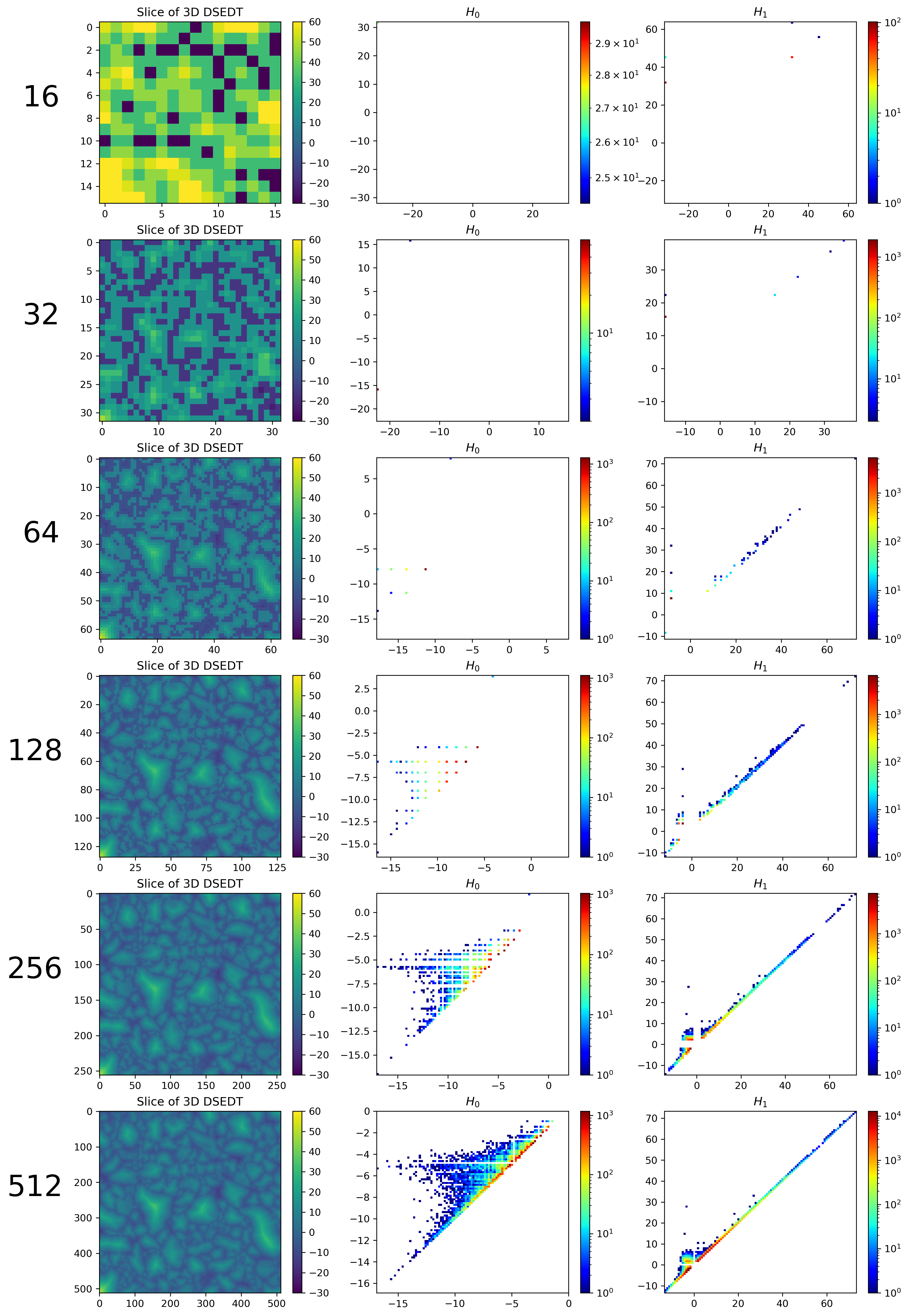}
    \caption{Each row shows a slice of the 3D SEDT, and the 0- and 1-dimensional persistence diagrams of the sand packing sample at different resolutions.}
    \label{fig:sandpack}
\end{figure*}

\end{document}